\documentclass{lmcs} 
\pdfoutput=1

\usepackage{lastpage}
\lmcsdoi{16}{4}{17}
\lmcsheading{}{\pageref{LastPage}}{}{}%
{Apr.~14,~2020}{Dec.~15,~2020}{}


\theoremstyle{plain} 

\usepackage{amsmath,amsfonts}
\usepackage{amssymb} 
\usepackage{ifxetex}
\usepackage[T1]{fontenc}
\usepackage[utf8]{inputenc}
\usepackage[cal=boondox]{mathalfa}
\usepackage{stmaryrd}
\usepackage{url}
\usepackage{tikz-cd}
\usepackage{amsthm}
\usepackage{mathpartir}

\newenvironment{diagram}{\begin{tikzcd}[sep=large]}{\end{tikzcd}}
\newenvironment{proofof}[1]
{\begin{proof}[Proof of {#1}]}
{\end{proof}}

\newcommand\sym[1]{\mathsf{#1}}

\newcommand{\cat}[1]{\mathcal{#1}}

\newcommand{\id}{\mathsf{id}}

\newcommand{\set}{\mathsf{Set}}

\newcommand{\inv}[1]{#1^{-1}}

\newcommand{\catT}{\mathbb{T}}

\newcommand{\R}{\mathsf{R}}
\newcommand{\RTy}{\mathsf{R_{Fam}}}
\newcommand{\RTm}{\mathsf{R_{El}}}
\newcommand{\transp}[1]{\overline{#1}}

\newcommand{\tick}{\mathsf{tick}}
\newcommand{\cminus}[2]{#1^{[#2-]}}

\newcommand{\opcat}[1]{{{#1}^{\mathrm{op}}}}

\newcommand{\triple}[3]{(#1;#2;#3)}
\newcommand{\timeobj}[2]{(#1;#2)}

\newcommand{\tickmap}[1]{\mathrm{tick}^{#1}}

\newcommand{\uniquemap}{!}

\newcommand{\FSA}{\mathcal{E}}

\renewcommand{\vartheta}{\delta}

\newcommand{\TobjA}{\triple{\FSA}\vartheta \lambda}
\newcommand{\TobjB}{\triple{\FSA'}{\vartheta'} {\lambda'}}
\newcommand{\TobjAminus}{\triple{\FSA}{\vartheta[\lambda-]}\lambda}
\newcommand{\TobjBminus}{\triple{\FSA'}{\vartheta'[\lambda'-]} {\lambda'}}

\newcommand{\grtotal}{\mathsf{GR}}
\newcommand{\gr}[1]{\mathsf{GR}[#1]}

\newcommand{\Set}{\mathsf{Set}}

\newcommand{\grstar}{\grtotal_\star}

\newcommand{\TStar}{\mathbb{T}_\star}
\newcommand{\TStarZ}{\mathbb{T}_\star^\ints}

\newcommand{\cwffont}[1]{\mathrm{#1}}

\newcommand{\typ}{\mathsf{Fam}}
\newcommand{\trm}{\mathsf{El}}
\newcommand{\cwftm}[3]{#1 \vdash #2 : #3}
\newcommand{\cwfty}[2]{#1 \vdash #2}

\newcommand{\p}[0]{\mathsf{p}}
\newcommand{\q}[0]{\mathsf{q}}
\newcommand{\compr}[2]{#1.#2}
\newcommand{\cpair}[2]{\la #1, #2\ra}

\newcommand{\cwfev}[2]{\mathsf{ev}(#1, #2)}

\newcommand{\idfam}[3]{\mathsf{Id}_{#1}(#2, #3)}

\newcommand{\Elsem}[1]{\mathcal{E}l^{#1}}
\newcommand{\Usem}[1]{\mathcal{U}^{#1}}

\newcommand{\code}[1]{\ulcorner #1 \urcorner}

\newcommand{\Usemin}[3]{\cwffont{in}_{#2, #3}^{#1}}

\newcommand{\clkfam}[1][\Gamma]{\clk_{#1}}

\newcommand{\clkel}{\chi}
\newcommand{\clkset}{\chi}
\newcommand{\clkmap}{\langle\chi\rangle}
\newcommand{\clkmapB}{\langle\chi'\rangle}

\newcommand{\dfixe}{\mathsf{dfix}}
\newcommand{\dfixebar}{\widehat{\mathsf{dfix}}}

\newcommand{\nxt}{\mathsf{next}}

\newcommand{\pair}[2]{\left(#1,#2\right)}

\newcommand{\Nat}[0]{\nats}

\newcommand{\la}[0]{\langle}
\newcommand{\ra}[0]{\rangle}

\newcommand{\dprod}[3]{\ensuremath{{\textstyle\prod\left(#1 : #2\right) . #3}}}

\newcommand{\subst}[2]{[#2/#1]}

\newcommand{\univ}[1]{\ensuremath{\operatorname{U}_{#1}}}
\newcommand{\elems}[1]{\ensuremath{\operatorname{El}_{#1}}}
\newcommand{\univin}[2]{\mathsf{in}_{#1,#2}}

\newcommand{\forallcode}[1]{\ensuremath{\overline{\forall}}}

\newcommand{\latbindcode}[3]{\overline{\triangleright}\, (#1:#2) . #3}
\newcommand{\latbindcodeAnn}[4]{\overline{\triangleright}_{#1}\, (#2:#3) . #4}

\newcommand{\timescode}{\overline{\times}}
\newcommand{\gStrcode}{\overline\gStr}
\newcommand{\natscode}{\overline\nats}

\newcommand{\syncode}[1]{\overline{#1}}

\newcommand{\predliftcode}[2][\kappa]{\overline{\predliftbare[#1]{#2}}}

\newcommand{\idty}[3]{#2 =_{#1} #3}
 
\newcommand{\PiApp}[5]{\mathsf{App}_{[#1:#2]#3}(#4,#5)}
\newcommand{\ForallApp}[4]{\mathsf{App}^\forall_{[#1]#2}(#3,#4)}
\newcommand{\ForallLam}[3]{\mathsf{Lam}^\forall_{[#1]#2} [#1]#3}
\newcommand{\TickApp}[5]{\mathsf{App}^\later_{[#1:#2]#3}(#4,#5)}
\newcommand{\TickLam}[4]{\mathsf{Lam}^\later_{[#1:#2]#3}[#1]#4}
\newcommand{\TickcApp}[5]{\mathsf{App}^\tickc_{[#2][#1]#3}([#2]#4,#5)}

\newcommand{\tickcsub}[3]{[(\tickc\! :\!  #3)/(#1 \! :\! #2)]}
\newcommand{\tickcsubsm}[3]{[(\tickc :  #3)/(#1  : #2)]}

\newcommand{\fc}[1]{\mathsf{fc}(#1)}

\newcommand{\later}{\triangleright}
\DeclareMathOperator{\tearlier}{\blacktriangleleft}
\DeclareMathOperator{\tlater}{\blacktriangleright}
\newcommand{\tlaterTy}[1][]{\tlater^{#1}_{\mathsf{Fam}}}
\newcommand{\tlaterTm}[1][]{\tlater^{#1}_{\mathsf{El}}}
\newcommand{\cv}{\mathrm{CV}}
\newcommand{\clk}{\mathrm{Clk}}

\newcommand{\dfix}{\mathsf{dfix}}
\newcommand{\fix}{\mathsf{fix}}
\newcommand{\pret}[1]{{\llbracket #1 \rrbracket}}
\newcommand{\den}[1]{{\llbracket #1 \rrbracket}}

\newcommand\tabs[2]{\lambda (#1 : #2).}
\newcommand\tickc{\diamond}
\newcommand\tapp[2][\tickA]{#2\,[#1] }
\newcommand\tappc[1]{\tapp[\tickc]{#1}}
\newcommand{\tickA}{\alpha}
\newcommand{\tickB}{\beta}
\newcommand\latbind[2]{{\triangleright}\, (#1:#2) .}
\newcommand\toksubst[3][\kappa]{\left[#2/#3\right]}

\newcommand{\clocktype}{\mathsf{clock}}

\newcommand\Str{\sym{Str}}
\newcommand\gStr[1][\kappa]{\sym{Str}^{#1}}

\newcommand{\cons}[2]{\mathrel{#1::#2}}

\newcommand{\predlift}[3][\kappa]{\gStr[#1]_{#2}(#3)}
\newcommand{\predliftbare}[2][\kappa]{\gStr[#1]_{#2}}

\newcommand{\basicsub}[2]{[#1\mapsto #2]}
\newcommand{\subex}[3]{#1\basicsub{#2}{#3}}

\newcommand{\wksubst}[2]{\mathsf{wk}_{#1;#2}}

\newcommand{\cirr}[1][\kappa]{\sym{cirr}^{#1}}
\newcommand{\tirr}[1][\kappa]{\sym{tirr}^{#1}}
\newcommand{\pfix}[1][\kappa]{\sym{pfix}^{#1}}

\newcommand{\nats}{\mathbb{N}}
\newcommand{\ints}{\mathbb{Z}}

\newcommand{\restrict}[2]{{#1}|_{#2}}

\newcommand{\ru}[2]{\dfrac{\begin{array}[b]{@{}c@{}} #1 \end{array}}{#2}}

\newcommand{\of}{{:}} 

\newcommand\hastype[4][]{
#2 \vdash_{#1} #3: #4
}
\newcommand\hasnotype[4][]{
#2 \vdash_{#1} #3
}
\newcommand\wfcxt[2][]{#2 \vdash_{#1}}
\newcommand\istype[3][]{
\ensuremath{#2 \vdash_{#1} #3 \, \operatorname{type}}
}

\newcommand\jeqjud[5][]{
#2 \vdash_{#1} #3 \jeq #4: #5
}

\newcommand{\jeq}{\equiv}
\newcommand{\jeqty}{\equiv}

\newcommand{\defeq}{\mathbin{\overset{\textsf{def}}{=}}}

\newcommand{\clott}{CloTT}
\newcommand{\gdtt}{GDTT}

\newcommand{\rasmus}[1]{}
\newcommand{\bassel}[1]{}
\newcommand{\niccolo}[1]{}

\begin{document}%
\renewcommand{\L}{\mathsf{L}}

\title[Ticking Clocks as Dependent Right Adjoints]{Ticking Clocks as Dependent Right Adjoints \\ 
   Denotational Semantics for Clocked Type Theory}

\author[B.~Mannaa]{Bassel Mannaa}	
\address{eToroX Labs, Denmark}	
\email{basselma@etorox.com}  
\thanks{This work was supported by DFF-Research Project 1 Grant no.
4002-00442, from The Danish Council for Independent Research for the Natural Sciences (FNU) and by a research grant (13156) 
from VILLUM FONDEN. Niccolò Veltri was also supported
by the ESF funded Estonian IT Academy research measure 
(project 2014-2020.4.05.19-0001).}	

\author[R.E.~M{\o}gelberg]{Rasmus Ejlers M{\o}gelberg}	
\address{Department of Computer Science, IT University of Copenhagen, Copenhagen, Denmark}	
\email{mogel@itu.dk}  

\author[N.~Veltri]{Niccolò Veltri}	
\address{Department of Software Science, Tallinn University of Technology, Tallinn, Estonia}	
\email{niccolo@cs.ioc.ee}  

\begin{abstract}
Clocked Type Theory (\clott) is a type theory for guarded recursion useful for programming with
coinductive types, allowing productivity to be encoded in types, and for reasoning about advanced 
programming language features using an abstract form of step-indexing. \clott\ has previously
been shown to enjoy a number of 
syntactic properties including strong normalisation, canonicity and decidability of the equational theory.
In this paper we present a denotational semantics for \clott\ useful, e.g., for studying future extensions 
of \clott\ with constructions such as path types. 

The main challenge for constructing this model is to model the notion of ticks on a clock used in \clott\ for coinductive
reasoning about coinductive types. We build on a category previously used to model guarded recursion
with multiple clocks. In this category there is an object of clocks but no object of ticks, and so 
tick-assumptions in a context can not be modelled using standard tools.
Instead we model ticks using dependent right adjoint functors, a generalisation of the category
theoretic notion of adjunction to the setting of categories with families. Dependent right adjoints are known to model 
Fitch-style modal types, but in the case of \clott, the modal operators constitute a family indexed internally in the type
theory by clocks. We model this family using a dependent
right adjoint on the slice category over the object of clocks. 
Finally, we show how to model the tick constant of \clott\ using a semantic substitution. 

This work improves on a previous model by two of the authors which not only had a flaw but was also
considerably more complicated. 
\end{abstract}

\maketitle

\section*{Introduction}

In recent years a number of extensions of Martin-L{\"o}f type theory (MLTT)~\cite{MartinLof:84} have been proposed to 
enhance the expressiveness or usability of the type theory. The most famous of these is Homotopy Type Theory~\cite{hottbook},
but other directions include the related Cubical Type Theory~\cite{CTT}, 
FreshMLTT~\cite{FreshMLTT}, a type theory with name abstraction
based on nominal sets, and Type Theory in Color~\cite{bernardy2015presheaf} 
for internalising relational parametricity in type theory. Many
of these extensions use denotational semantics to argue for consistency and to inspire constructions in the language.

This paper is part of a project to extend type theory with guarded recursion~\cite{Nakano:Modality}, a variant of 
recursion that uses a modal type operator $\later$ (pronounced `later') 
to preserve consistency of the logical reading of type theory.
The type $\later A$ should be read as classifying data of type $A$ available one time step from now, and comes with
a map $\nxt : A \to \later A$ and a fixed point operator mapping a function $f : \later A \to A$ to a fixed point for 
$f \circ \nxt$. This, in combination with \emph{guarded recursive types}, i.e., types where the recursion variable
is guarded by a $\later$, e.g., $\gStr[\sym{g}] \jeqty \nats \times \later \gStr[\sym{g}]$ gives a powerful type theory in which 
operational models of combinations of advanced programming language features such as higher-order
store~\cite{Birkedal-et-al:topos-of-trees} and nondeterminism~\cite{Bizjak-et-al:countable-nondet-internal} can be modelled
using an abstract form of step-indexing~\cite{Appel:M01}. 
Combining guarded recursion with a notion of clocks, indexing the $\later$ operator with
clock names, and universal quantification over clocks, one can encode coinduction using guarded recursion,
allowing productivity~\cite{coquand1993infinite} of coinductive definitions to be encoded in types~\cite{atkey13icfp}. 
For example, if $\gStr[\kappa]$ is a type of streams guarded on the clock $\kappa$, i.e., satisfying the equation
$\gStr[\kappa] \jeqty \nats \times \later^\kappa \gStr[\kappa]$, then the type $\Str \defeq \forall\kappa . \gStr[\kappa]$ 
obtained by universally
quantifying the clock $\kappa$ is a coinductive type of streams satisfying the more standard type isomorphism
$\Str\cong \nats \times \Str$.

The most advanced type theory with all the above mentioned features is Clocked Type Theory (\clott)~\cite{bahr2017clocks}, 
which introduces the notion of ticks 
on a clock. Ticks are evidence that time has passed and can be used to unpack elements of type $\later^\kappa A$ to elements
of $A$. In fact, in \clott\ $\later^\kappa A$ is generalised to a special form of dependent function type 
$\latbind\tickA\kappa A$ from ticks to $A$. 
The introduction rule abstracts assumptions of the form $\tickA : \kappa$ from the context, and the elimination applies a 
term $t : \latbind\tickA\kappa A$ to a tick $\tickB : \kappa$ to give an element of $A \subst\tickA\tickB$. 
Special typing rules ensure that a term is never applied twice to the same tick. The combination of ticks and clocks
in \clott\ can be used for coinductive reasoning about coinductive types, by encoding the \emph{delayed substitutions} 
of~\cite{GDTT}. 

Bahr et al~\cite{bahr2017clocks} have shown that \clott\ can be given a reduction semantics satisfying 
strong normalisation, confluence and
canonicity. This establishes that productivity can indeed be encoded in types: For a closed term $t$ of stream type, 
the $n$'th element can be computed in finite time. These syntactic results also imply soundness of the type theory. 
However, these results have only been established for a core type theory without, e.g., identity types, and
the arguments can be difficult to extend to larger calculi. In particular, we are interested in extending \clott\ with path
types as in Guarded Cubical Type Theory~\cite{GCTT} in future work. 
Therefore a denotational model of \clott\  can be useful, and this paper presents such a model. 

The work presented here builds on a number of existing models for guarded recursion. The most basic such, modelling the
single clock case, is the topos of trees model~\cite{Birkedal-et-al:topos-of-trees}, 
in which a closed type is modelled as a family of sets 
$X_n$ indexed by natural numbers $n$, together with restriction maps of the form $X_{n+1} \to X_n$ for every $n$. In other
words, a type is a presheaf over the ordered natural numbers. In this model $\later$ is modelled as $(\later X)_0 = 1$ and 
$(\later X)_{n+1} = X_n$ and guarded recursion reduces to natural number recursion. The guarded recursive 
type $\gStr[\sym{g}]$ mentioned above can be modelled in the topos of trees as 
$\gStr[\sym{g}]_n = \nats^{n+1}\times 1$. 

Bizjak and M{\o}gelberg~\cite{GDTTmodel} recently extended this model to the case of many clocks, using a category
$\set^\catT$ of covariant presheaves over a category $\catT$ of time objects. An object of $\catT$ is a pair 
of a finite set $\FSA$ and a map 
$\vartheta : \FSA \to \nats$, and a morphism from $(\FSA, \vartheta)$ to $(\FSA', \vartheta')$ is a map 
$\tau : \FSA \to \FSA'$ such that $\vartheta' \tau \leq \vartheta$ in the pointwise order. Intuitively, $\FSA$
indicates the set of clocks in play at any time in a computation, and $\vartheta$ indicates the number of ticks
left on each clock. The use of the inequality in the maps allows for time to pass, similarly to the passing from
a larger number to a smaller number in the topos of trees model. 

%
The main challenge when adapting the model of~\cite{GDTTmodel} to \clott\ is to model ticks, which were not present
in the language modelled in~\cite{GDTTmodel}. In particular, how does one model tick assumptions of the form
$\tickA : \kappa$ in a context, when there appears to be no object of ticks in the model to be used as the
denotation of the clock $\kappa$? In this paper we observe that these assumptions can be modelled using a left adjoint
$\tearlier^\kappa$ to the functor $\tlater^\kappa$ used in~\cite{GDTTmodel} to model the delay modality $\later^\kappa$
associated to the clock $\kappa$. Precisely we model context extension as $\pret{\Gamma, \tickA : \kappa} = 
\tearlier^\kappa\pret\Gamma$. The modality $\latbind\tickA\kappa A$ is then modelled as a \emph{dependent right adjoint}
to $\tearlier^\kappa$, 
a notion studied in detail in~\cite{drat}: If $\cat C$ is a category with family (CwF)~\cite{dybjer1996} (a standard notion of model 
for dependent type theory) and $L$ an endofunctor on (the underlying category of) $\cat C$, a dependent right 
adjoint to $L$ is an operation mapping a family $A$ over $L\Gamma$ to a family $RA$ over $\Gamma$ with a bijective
correspondence between elements of $A$ and elements of $RA$ natural in $\Gamma$. Dependent right adjoints model 
Fitch-style modal operators in type theory, a general pattern seen also in the model of fresh name abstraction of 
FreshMLTT~\cite{FreshMLTT} and dependent path types in cubical type theory~\cite{CTT}. 
In \clott\ the type operator $\later$ is indexed by clocks, and since the model has an object of clocks
this can be understood as an internally indexed family of Fitch-style modal operators. We show how to model this as a dependent
right adjoint on the slice category over the object of clocks.

Finally we show how to model the special tick constant $\tickc$ used in \clott\ to eliminate $\later^\kappa$ in special 
situations. Again, since there is no object of ticks in which $\tickc$ can be an element, standard tools can not be used to model this.
Still, we shall see that there exists a semantic substitution of $\tickc$ for a tick variable
that can be used to model application of terms to $\tickc$. 

\subsection*{Overview}

Before introducing Clocked Type Theory in full 
we focus on a fragment called the \emph{tick calculus} capturing just the interaction of ticks with dependent types. 
Section~\ref{section:tick:calc} introduces this and shows how ticks can be used to program with and reason about
modal types. Then we introduce the notion of dependent right adjoint and show how to use this to model the tick
calculus. Section~\ref{sec:clott} introduces \clott\ as an extension of the tick calculus to multiple clocks 
and with guarded recursion. In the original presentation of \clott~\cite{bahr2017clocks} judgements had a separate
context of clock variables. Here we use a single context, and this simplifies not only the syntax but also the semantics
considerably. Section~\ref{sec:syntax:universes} extends basic \clott\ with universes following the approach of 
Guarded Dependent Type Theory~\cite{GDTTmodel}. For universes to be consistent with the \emph{clock irrelevance axiom}
of \clott\ these are indexed by sets of clocks that may appear freely in the elements of the universe. Inclusions between
sets of clocks induce inclusions between universes and all type constructors commute on the nose with these. 

Section~\ref{sec:basic:model} introduces the presheaf category $\grtotal$ 
forming the model of \clott\ and defines the object of clocks in this. 
This is the same category as used by Bizjak and M{\o}gelberg~\cite{GDTTmodel} to model the related Guarded Dependent Type Theory,
and it was also discovered independently by Harper and Sterling~\cite{sterling2018guarded} as a model of
Guarded Computational Type Theory. Section~\ref{sec:dradjoint} constructs a dependent right adjoint on the slice category 
over the object of clocks, and Section~\ref{sec:modelling:ticks} lifts these results to an internally indexed family of dependent
right adjoints on $\grtotal$. Sections~\ref{sec:guarded:rec} and~\ref{sec:tickc} describe the semantic structure required
to model the guarded fixed point operator and the tick constant $\tickc$, respectively. Section~\ref{sec:semantic:universes} recalls the
semantic universes of~\cite{GDTTmodel} and shows how to model the modal types of \clott\ in these. 

Section~\ref{sec:interp:syntax} defines the interpretation of syntax into the model and proves soundness. For this we
follow the approach of Hofmann~\cite{Hofmann1997} for modelling dependent type theories: First the interpretation
of syntax is defined as a partial function, then it is proved that the interpretation is defined for all judgements that have
a derivation. The latter proof is done by a simultaneous induction with proofs of soundness and a substitution lemma. 
As is standard, the syntax interpreted into the model is an annotated variant of the syntax presented in Section~\ref{sec:clott}.
Apart from the standard annotations e.g. of application terms with the $\Pi$-type of the function, in \clott\ the term 
for application to the tick constant $\tickc$ must be changed by replacing a substitution by an explicit substitution. Moreover, 
special lemmas for weakening substitutions must be proved to accommodate tick-weakening in \clott. The paper ends
with conclusions and future work in Section~\ref{sec:conclusion}. 

\subsection*{Related work}

The two first named authors have previously published a conference publication~\cite{conferenceversion} 
describing a model of \clott. That paper contained an error in the description of 
the left adjoint $\tearlier^\kappa$, which had consequences for a number of other results in the paper. 
Apart from correcting this mistake the present paper also presents a greatly simplified model construction. 
The previous model used the original syntax of \clott\ in which judgements had a separate context 
of clock variables $\Delta$, and modelled this using a diagram of categories $\gr\Delta$ indexed by
clock contexts. These categories were equivalent to slice categories of the category $\grtotal$ used in this
paper, and are also used in Section~\ref{sec:semantic:universes} to construct the semantic universes. The 
clock contexts $\Delta$ allowed the modal operators to be externally indexed. In particular, each $\kappa \in \Delta$
induced a dependent right adjoint on $\gr\Delta$. Unfortunately, the morphisms of the diagram induced
by clock substitutions did not commute with the left adjoints of these dependent adjointions causing great complications of
the model construction. The present paper avoids these problems by using an internal indexing of the 
dependent adjunctions. 
%

As described above, one of the motivations for \clott\ is the encoding of coinductive types capturing
the notion of productivity in types. There exist other solutions to this problem, in particular the combination of 
single clock guarded recursion with an `always' modality $\Box$~\cite{birkedal2017guarded,gratzer2020multimodal}
and sized types~\cite{HughesPS96,Abel:Wellfounded,Abel:NBE:sized:types,Sacchini13}. We refer to~\cite{GDTTmodel}
for a discussion of the relationship between these approaches.

\section{A tick calculus}
\label{section:tick:calc}

Before introducing \clott\ we focus on a fragment to explain the notion of ticks and how to model these. To motivate 
ticks, consider the notion of applicative functor from functional programming~\cite{mcbride2008applicative}: 
a type former $\later$ with maps
$A \to \later A$ and $\later(A \to B) \to \later A \to \later B$ satisfying a number of equations that we shall not recall. 
These maps can be used 
for programming with the constructor $\later$, but for reasoning in a dependent type theory, one needs an extension of these
to dependent function types. 
For example, in guarded recursion one can prove a theorem $X$ by constructing a map $\later X \to X$ and taking its fixed
point in $X$. If the theorem is that a property holds for all elements in a type of guarded streams satisfying 
$\gStr[] \jeqty \nats \times \later \gStr[]$, then $X$ will be of the form $\dprod{xs}{\gStr[]}P$. To apply the 
(essentially coinductive) assumption of type $\later\dprod {xs}{\gStr[]}P$ to the tail of a stream, which has type
$\later\gStr[]$ we need an extension of the applicative functor action. 

What should the type of such an extension be? Given $a: \later A$ and $f: \later (\dprod xAB)$ 
the application of $f$ to $a$ should be something of the form $\later B\subst x{??}$. If we think of $\later$ as a delay,
intuitively $a$ is a value of type $A$ delayed by one time, and the $??$ should be the value delivered
by $a$ one time step from now. One way of referring to that value is by changing the 
target type of the dependent applicative functor action to a $\mathsf{let}$-expression. 
Here we describe a more direct approach based on ticks. Ticks should be though of as 
evidence that time has passed which can be used to unpack elements of modal type.

 
The \emph{tick calculus} is the extension of dependent type theory with the following four rules
\begin{mathpar}
\inferrule*{\wfcxt[]{\Gamma}\and \tickA \notin \Gamma}{\wfcxt[]{\Gamma, \tickA : \tick}} \and
\inferrule*{\istype[]{\Gamma, \tickA : \tick}{A}}{\istype[]{\Gamma}{\latbind\tickA\tick A}} \\
\inferrule*{\hastype[]{\Gamma, \tickA : \tick}t{A}}{\hastype[]{\Gamma}{\tabs\tickA\tick t}{\latbind\tickA\tick A}} 
\and 
\inferrule*{\hastype[]{\Gamma}t{\latbind\tickA\tick A}\and  \and \wfcxt[]{\Gamma,\tickB : \tick, \Gamma'}}{\hastype[]{\Gamma,\tickB : \tick, \Gamma'}{\tapp[\tickB]t}{A[\tickB/\tickA]}} 
\end{mathpar}
An assumption of the form $\tickA : \tick$ in a context is an assumption that one time step has passed, and $\tickA$ 
is the evidence of this. Variables on the right-hand side of such an assumption should be thought of as arriving one time step later than those on the left. Ticks can be abstracted in terms and types, so that the type constructor $\later$ now comes
with evidence that time has passed that can be used in its scope. The type $\latbind\tickA\tick A$ can be thought
of as a form of dependent function type over ticks, which we abbreviate to $\later A$ if $\tickA$ does not occur
free in $A$. 
The elimination rule states that if a term $t$ can be typed as $\latbind\tickA\tick A$ before the arrival 
of tick $\tickB$, $t$ can be opened using $\tickB$ to give an element of type $A[\tickB/\tickA]$. Note that the 
causality restriction in the typing rule prevents a term like 
$\lambda x .  \tabs\tickA\tick \tapp{\tapp x} : \later\later A \to \later A$ being well typed; a tick can only 
be used to unpack the same term once. The context $\Gamma'$ in the elimination rule ensures that typing rules are
closed under weakening, also for ticks. 
Note that the clock object $\tick$ is not a type. The variable introduction rule is unchanged: 
$\hastype[]{\Gamma, x: A, \Gamma'}xA$ even if there are ticks in $\Gamma'$. Intuitively, this means that data
is kept past time steps.

The equality theory is likewise extended with the usual $\beta$ and $\eta$ rules: 
\begin{align*}
 (\tabs\tickA\tick \tapp[\tickB]{t)} & = t[\tickB/\tickA] & \tabs\tickA\tick (\tapp t) & = t
\end{align*}
As stated, the tick calculus should be understood as an extension of standard dependent type theory. In particular one 
can add dependent sums and function types with standard rules. 

We can now type the dependent applicative structure as 
\begin{align*}
 \lambda( x \of A). \tabs\tickA\tick x & \,\of\, A \to \later A \\
 \lambda f  .\lambda y.  \tabs\tickA\tick  \tapp f(\tapp y) &\,\of\, 
 \later\left(\dprod xAB\right) \to \dprod y{\later A}{\latbind\tickA\tick {B\subst x{\tapp y}}}
\end{align*}

\begin{exa} \label{ex:streams}
For a small example on how ticks in combination with the fixed point operator 
$\dfix : (\later X \to X) \to \later X$ 
can be used to reason about guarded recursive data, let 
$\gStr[] \jeqty \nats \times \later \gStr[]$ be the type of guarded recursive streams mentioned above, and
suppose $x\of \Nat \vdash P(x)$ is a family to be thought of as a predicate on $\Nat$. A lifting of $P$ to streams
would be another guarded recursive type $y\of \Str \vdash \predlift[]{P}y$ satisfying 
$\predlift[]{P}{\cons x{xs}} \jeqty P(x) \times \latbind\tickA\tick {\predlift[]{P}{xs\,[\alpha]}}$ (where 
$\cons x{xs}$ is the pairing of $x$ and $xs$). If $p: \dprod x\Nat {P(x)}$ is a proof of $P$ 
we would expect that also $\dprod y \Str {\predlift[]{P}y}$ can be proved, and indeed this can be done as follows. 
Consider first 
\begin{align*}
 f & : \later (\dprod y \Str {\predlift[]{P}{y}}) \to \dprod y\Str{\predlift[]{P}y} \\
 f\, q \, (\cons x{xs}) & \defeq \pair{p(x)}{\tabs\tickA \tick \tapp[\tickA]q(\tapp[\tickA]{xs})}
\end{align*}
Then $f(\dfix(f))$ has the desired type.  
\end{exa}

More generally, ticks can be used to encode \cite{bahr2017clocks} the \emph{delayed substitutions} of \cite{GDTT}, 
which have been used to reason coinductively about coinductive data. For more examples of 
reasoning using these see~\cite{GDTT}. 
The tick calculus is an example of a Fitch-style modal 
calculus~\cite{clouston2018fitch,Fitch:Symbolic}. Most of these use a presentation 
where ticks are simply markers in the context, rather than carry names as here. However, 
names of ticks play a crucial role in the normalisation proof for \clott~\cite{bahr2017clocks},
and we therefore also use names here. 

%
%

\subsection{Modelling ticks using adjunctions}
\label{externalcwf}


We now describe a notion of model for the tick calculus. It is based on the notion of category with family (CwF) \cite{dybjer1996},
which is a standard notion of model of dependent type theory. 
\begin{defi}
 A CwF comprises
 \begin{itemize}
\item A category $\cat C$ with a distinguished terminal object
\item For each object $\Gamma$ of $\cat C$ a set $\typ(\Gamma)$ of \emph{families} over $\Gamma$. 
We write $\cwfty\Gamma A$ to mean $A\in \typ(\Gamma)$.
\item For each $\Gamma$ in $\cat C$ and each family $A$ in $\typ(\Gamma)$ a set $\trm(A)$
of \emph{elements} of $A$. We write $\cwftm\Gamma tA$ to mean $t\in \trm(A)$.
\item For each morphism $\gamma : \Delta \to \Gamma$ in $\cat C$ \emph{reindexing} operations
mapping $\cwfty\Gamma A$ to $\cwfty\Delta{A[\gamma]}$ and $\cwftm\Gamma tA$ to $\cwftm\Delta{t[\gamma]}{A[\gamma]}$. 
These must
satisfy the equations $A[\id] = A$, $t[\id] = t$,  
$A[\gamma\circ \delta] = A[\gamma][\delta]$ and $t[\gamma\circ \delta] = t[\gamma][\delta]$
for all morphisms $\delta$ with codomain $\Delta$.  
\item A \emph{comprehension} operation associating to each family $\cwfty{\Gamma}A$ 
the following: An object $\compr \Gamma A$ in $\cat C$, a morphism 
$\p_A : \compr \Gamma A \to \Gamma$ and an element $\cwftm{\compr\Gamma A}{\q_A}{A[\p_A]}$, such that for every
$\gamma : \Delta \to \Gamma$, and $\cwftm\Delta t{A[\gamma]}$ there exists a unique morphism 
$\cpair \gamma t : \Delta \to \compr \Gamma A$ such that  $\p_A \circ\cpair \gamma t = \gamma$ and 
$\q_A[\cpair \gamma t] = t$. 
\end{itemize}
\end{defi}

The requirements on reindexing of families and elements mean that they can be described more 
concisely as a functor from $\opcat{\cat C}$ to the category of families on sets. This is in fact Dybjer's original
definition. Awodey's \emph{natural models of type theory}~\cite{awodey2018natural} 
are an elegant abstract formulation of the notion of CwF. 

To model the tick calculus we need an operation $\L$ modelling the extension of a context with a tick, plus an 
operation $\R$ modelling
$\later$. In the simply typed setting, $\R$ would be a right adjoint to context extension modelling the bijective
correspondence between terms $\hastype[]{\Gamma, \tickA : \tick}{t}A$ and terms 
$\hastype[]{\Gamma}{s}{\latbind\tickA\tick A}$. For dependent types this is 
not quite so, since these operations work on different objects (contexts and types respectively). Instead, we need a
dependent adjunction as in the following definition, which generalises that of~\cite{drat} by allowing also
dependent adjunctions between different categories (not just endoadjunctions). 
\begin{defi} \label{def:dep:adj}
 Let $\cat C$ and $\cat D$ be CwFs and let $\L : \cat C \to \cat D$ be a functor between 
 the underlying categories. A \emph{dependent right adjoint} to $\L$ consists of an operation associating 
 to each family $\cwfty{\L\Gamma}A$ 
 in $\cat D$ a family $\cwfty{\Gamma}{\R A}$ in $\cat C$ and a bijective map of elements mapping  
 $\cwftm{\L \Gamma} tA$ to $\cwftm{\Gamma} {\transp{t}}{\R A}$ such that 
 $(\R A)[\gamma] = \R (A[\L\gamma])$ and $\transp t[\gamma] = \transp{t[\L \gamma]}$.
\end{defi}
We write $\transp{(-)}$ also for the inverse direction of the bijection on terms so that $\transp{\transp{t}} = t$. It easily 
follows~\cite{drat}
that also the inverse direction commutes with substitution, i.e., that for $\cwftm{\Gamma}{s}{\R A}$ 
also $\transp{s[\gamma]} = \transp s[\L\gamma]$. 

The dependent adjunctions in this paper arise from adjunctions on the underlying categories with 
liftings of the right adjoint to families and elements as in the following definition. 
\begin{defi}
\label{def:cwfa}
 Let $\cat{C}$ and $\cat{D}$ be CwFs and let $\R : \cat C \to \cat D$ be a functor. An \emph{extension of $\R$ to families 
 and elements}
 is a pair of operations presented here in the form of rules
 \begin{gather*}
\ru{\cwfty\Gamma A}{\cwfty{\R \Gamma}{\RTy A}} \qquad \ru{\cwftm\Gamma tA}{\cwftm{\R \Gamma}{\RTm t}{\RTy A}} 
\end{gather*}
commuting with reindexing in the sense that $(\RTy A)[\R \gamma] = \RTy (A[\gamma])$ and $(\RTm t)[\R \gamma] 
= \RTm (t[\gamma])$
hold for all substitutions $\gamma$, and commuting with comprehension in the sense that 
$\cpair{\R \p_A}{\RTm \q_A} : \R(\Gamma.A) \to \R\Gamma.\RTy A $ is an isomorphism.
\end{defi}

\begin{lem}
\label{lem:bijectivecorresp}
Let $\cat C$ and $\cat D$ be CwFs and let $\L : \cat{C} \to \cat{D} : \R$ be an adjunction of the underlying categories, such that $\R$ extends to families and elements. Let $\eta$ be the unit
and let $\epsilon$ be the counit of the adjunction. 
The operation mapping $\cwfty{\L\Gamma}{A}$ to $\cwfty{\Gamma}{\R A}$ defined as $\R A \defeq (\RTy A)[\eta]$ defines a dependent right adjoint to $\L$. The required bijection 
on elements maps $\cwftm{\L\Gamma}aA$ to $(\RTm a)[\eta]$ and 
$\cwftm{\Gamma}b{\R A}$ to 
$\q_A[\epsilon \circ \L (\inv{\cpair{\R \p_A}{\RTm \q_A}}\circ \la \eta,b\ra)]$.
\end{lem}

Lemma~\ref{lem:bijectivecorresp} is a straight-forward generalisation of~\cite[Lemma~17]{drat}. 
Note the notational convention: In the setting 
of the lemma we overload $\R$ for both the functor on the underlying category and the dependent right adjoint, and use the
more verbose $\RTy$ and $\RTm$ for the extension of $\R$ to families and elements. This differs from the notation
used in~\cite{drat}, but is chosen here for notational convenience. 

%
%
%
%
%
%

\subsection{Interpretation}
\label{sec:tick:calc:interp}

The tick calculus can be modelled in a CwF equipped with an endofunctor $\L$ with a dependent right adjoint and a natural
transformation $\p_\L : \L \to \id_{\cat C}$. The latter is needed to interpret tick weakening. Defining 
\[
\pret{\wfcxt[]{\Gamma,\alpha : \tick}} = \L \pret{\wfcxt[]{\Gamma}} 
\]
$\p_\L$ allows us to define a context projection $\p_{\Gamma'} : \pret{\Gamma, \Gamma'\vdash} \to \pret{\Gamma\vdash}$ 
by induction on $\Gamma'$ using $\p_\L$ in the case of tick variables. We can then define the rest of the interpretation as
\begin{align*}
 \pret{\hasnotype[]{\Gamma, x : A, \Gamma'}xA} & = \q_A[\p_{\Gamma'}] & 
 \pret{\istype[]{\Gamma}{\latbind\tickA\tick A}} & = \R\pret{A} \\
\pret{\hasnotype[]{\Gamma}{\tabs\tickA\tick t}{\latbind\tickA\tick A}} & =  \transp{\pret{t}} & 
\pret{\hasnotype[]{\Gamma,\tickB : \tick,\Gamma'}{\tapp[\tickB]t}{A\subst\tickA\tickB}} & =  \transp{\pret{t}}[ \p_{\Gamma'}]
\end{align*}

\begin{prop}\label{prop:tick:soundness}
 The above interpretation of the tick calculus into a CwF with adjunction and tick weakening $\p_\L$ is sound. 
\end{prop}

Proposition~\ref{prop:tick:soundness} can be proved using the tools of~\cite{drat}. 

\subsection{Adding basic type constructors}

The model of the tick calculus can be extended with basic type constructors like natural numbers, 
$\Pi$- and $\Sigma$-types as well as identity types. Here we just recall what it means for a 
CwF to have extensional identity types, referring the reader to Hofmann~\cite{Hofmann1997} for details on other
constructors.

\begin{defi}
 A CwF $\cat C$ has \emph{extensional identity types} if for each pair of elements 
 $\cwftm\Gamma tA$ and $\cwftm\Gamma uA$ of the same family $\cwfty\Gamma A$
 there is a family $\cwfty\Gamma{\idfam Atu}$ with at most one element such that 
 $t$ and $u$ are equal if and only if there is an element of $\cwfty\Gamma{\idfam Atu}$,
 and such that $(\idfam Atu)[\gamma] = \idfam{A[\gamma]}{t[\gamma]}{u[\gamma]}$.
\end{defi}

%
%
%
%
%

\section{Clocked Type Theory}
\label{sec:clott}

Clocked Type Theory (\clott) is an extension of the tick calculus with guarded recursion and multiple clocks. Rather than
having a global notion of time as in the tick calculus, ticks are associated with clocks and clocks can be assumed 
and universally quantified. In the original presentation of \clott~\cite{bahr2017clocks} judgements had a separate 
context for clock variables, i.e., assumptions of the form $\kappa : \clocktype$. In this paper, clock variables are simply assumed in the 
context as if they were ordinary variables. This simplifies both the syntax and semantics of the language. There are
no operations for forming clocks, only clock variables. It is often convenient to have a single clock constant $\kappa_0$ and this
can be achieved by a precompilation adding $\kappa_0$ as a fresh variable to the contexts.

\begin{figure*}[tbp]
\begin{center}
\textbf{Context formation rules} 
\begin{mathpar}
\inferrule*
{}
{\wfcxt{\cdot}}
\and
\inferrule*
{\istype{\Gamma}A \and x\notin \Gamma}
{\wfcxt{\Gamma, x : A}}
\and
\inferrule*
{\wfcxt{\Gamma} \and \kappa\notin \Gamma}
{\wfcxt{\Gamma, \kappa : \clocktype}}
\and
\inferrule*
{\hastype{\Gamma}\kappa\clocktype \and \tickA\notin \Gamma}
{\wfcxt{\Gamma, \tickA : \kappa}}
\end{mathpar}
\textbf{Type formation rules} 
\begin{mathpar}
  \inferrule*
  {\istype{\Gamma,\tickA:\kappa}{A}}
  {\istype{\Gamma}{\latbind{\tickA}{\kappa} A}}
  \and
  \inferrule*
  {\istype{\Gamma,\kappa : \clocktype}{A}}
  {\istype{\Gamma}{\forall \kappa . A}}
  \and
  \inferrule*
  {\wfcxt{\Gamma}}
  {\istype{\Gamma}{\nats}}
\end{mathpar}
\textbf{Typing rules}
\begin{mathpar}
  \inferrule*
  {\hastype{\Gamma,\kappa : \clocktype}{t}{A}}
  {\hastype{\Gamma}{\Lambda\kappa. t}{\forall \kappa . A}}
  \and
  \inferrule*
  {\hastype{\Gamma}{t}{\forall \kappa . A}\\
    \hastype\Gamma{\kappa'}\clocktype}
  {\hastype{\Gamma}{t [\kappa']}{A \subst{\kappa}{\kappa'}}}
  \and
  \inferrule*
  {\hastype{\Gamma,\tickA:\kappa}{t}{A}}
  {\hastype{\Gamma}{\tabs{\tickA}{\kappa} t}{\latbind{\tickA}{\kappa} A}}
  \and
  \inferrule*
  {\hastype{\Gamma}{t}{\latbind{\tickA}{\kappa} A}\\\wfcxt{\Gamma,\tickA':\kappa,\Gamma'}}
  {\hastype{\Gamma,\tickA': \kappa,\Gamma'}{\tapp[\tickA'] t}{A\toksubst{\tickA'}{\tickA}}}
  \and
  \inferrule*
  {\hastype{\Gamma,\kappa : \clocktype}{t}{\latbind{\tickA}{\kappa} A}\\ \hastype\Gamma{\kappa'}{\clocktype}}
  {\hastype{\Gamma}{\tappc{(t\,\subst{\kappa}{\kappa'})}}{A\subst{\kappa}{\kappa'}\toksubst{\tickc}{\tickA}}} 
  \and
  \inferrule*
  {\hastype{\Gamma}{t}{\later^\kappa A \to A}}
  {\hastype{\Gamma}{\dfix^\kappa\,t}{\later^\kappa A}}
  \and
  \inferrule*
  {\hastype{\Gamma}{t}{A} \\ A \jeq B}
  {\hastype{\Gamma}{t}B}
  \and
  \inferrule*
  {
  \kappa : \clocktype \in \Gamma}
  {\hastype{\Gamma}\kappa\clocktype}
  \and
  \inferrule*
  {
   x : A \in \Gamma}
  {\hastype{\Gamma}xA}
\end{mathpar}
\textbf{Judgemental equality}
  \begin{align*}
    (\Lambda\kappa.t)[\kappa']&\jeq  t \subst\kappa{\kappa'} &
    \Lambda \kappa. (t [\kappa]) &\jeq t & 
   \tapp[\tickB]{(\tabs{\tickA}{\kappa} t)} &\jeq t\toksubst{\tickB}{\tickA} \\ 
    \tabs{\tickA}{\kappa} (\tapp t) &\jeq t & 
    \tappc{(\dfix^\kappa \, t)} &\jeq t\,(\dfix^\kappa\,t) &
    \tappc{(\tabs{\tickA}{\kappa} t)} &\jeq t\toksubst{\tickc}{\tickA} 
  \end{align*}
\caption{Selected typing and judgemental equality rules of Clocked Type Theory. The two $\eta$ rules are subject
to the standard conditions of $\kappa$ and $\tickA$, respectively, not appearing in the term $t$.}
\label{fig:clott:typing}
\end{center}
\end{figure*}

The rules for typing judgements and judgemental equality are given in Figure~\ref{fig:clott:typing}. These should 
be seen as an extension of a dependent type theory with $\Pi$- and $\Sigma$-types, as well as 
extensional identity types. The rules for these are completely standard, and thus are omitted from
the figure. We write $\jeq$ for judgemental equality and $\idty Atu$ for identity types. The model will also model the 
\emph{identity reflection} rule
\[
\inferrule*{\hastype{\Gamma}{p}{\idty Atu}}
{\jeqjud{\Gamma}tuA}
\]
of extensional type theory. 


The type of the guarded fixed point operator $\dfix$ uses the abbreviation $\later^\kappa A$
for $\latbind{\tickA}{\kappa} A$ where $\tickA$ does not occur free in $A$. This operator is 
useful in combination with guarded recursive types such as 
a type of guarded streams $\gStr$ satisfying $\gStr\jeq \nats \times \later^\kappa \gStr$. This type is similar to the one from 
Example~\ref{ex:streams} except that the delay now is 
associated with a clock variable $\kappa$. We will see how to define
such guarded recursive types in the next section. Given $\gStr$ we can
use $\dfix$ for recursive programming with guarded streams, e.g., when defining a constant stream of zeros as
$\dfix^\kappa (\lambda x. \cons 0x)$. The type of $\dfix$ ensures that only productive recursive definitions are typeable, e.g.,
$\dfix^\kappa (\lambda x . x)$ is not. 

The tick constant $\tickc$ gives a way to execute a delayed computation $t$ of type $\later^\kappa A$ to compute a value of
type $A$. In particular, if $t$ is a fixed point, application to the tick constant unfolds the fixed point once. This explains the need
to name ticks in \clott: Substitution of $\tickc$ for a tick variable $\tickA$ in a term allows for all fixed points applied to $\tickA$ 
in the term to be unfolded. In particular, the names of ticks are crucial for the strong normalisation result for \clott\ in \cite{bahr2017clocks}. 

Intuitively $\tickc$ is a constant of type $\kappa$ for any clock $\kappa$. However, since 
clocks are not types, $\tickc$ can only be introduced by applying it to a term of type 
$\latbind \tickA\kappa A$, and such applications must moreover be restricted to ensure
productivity. In particular a term such as 
$\lambda x : \later^\kappa A . \tappc x$ should not be well typed, as this would give a way 
of inhabiting all types using $\dfix$. 
The typing rule for application to the tick constant 
ensures this by assuming that the clock $\kappa$ associated to the delay does not occur freely 
in the type of any other variable in the context of $t$. For example,
the rule 
\[
  \inferrule*
  {\hastype{\Gamma, \kappa : \clocktype}{t}{\latbind{\tickA}{\kappa} A}}
  {\hastype{\Gamma, \kappa : \clocktype}{\tappc{t}}{A\toksubst{\tickc}{\tickA}}} 
\]
is admissible, which can be proved using a weakening lemma. This rule, however, is not closed
under variable substitution, which is the motivation for the more general rule of Figure~\ref{fig:clott:typing}. The typing
rule is a bit unusual, in that it involves substitution in the term in the conclusion. In the elaborated syntax for \clott\ to be interpreted
in the model in Section~\ref{sec:interp:syntax}, this substitution is replaced by an explicit substitution binding $\kappa$ in $t$ rather 
than substituting it away. 
%
%

Universal quantification over clocks allows for coinductive types to be encoded using guarded recursive types~\cite{atkey13icfp}. 
For example
$\Str \defeq \forall\kappa . \gStr$ is a coinductive type of streams. The head and tail maps $\sym{hd} : \Str \to \nats$ and 
$\sym{tl} : \Str \to \Str$ can be defined as
\begin{align*}
 \sym{hd}(xs) & \defeq \pi_1(xs[\kappa_0]) &
 \sym{tl}(xs) & \defeq  \Lambda \kappa . (\tappc{(\pi_2(xs[\kappa]))})
\end{align*}
using the clock constant $\kappa_0$. It is easily seen that 
$\Str \jeq \forall\kappa . (\nats \times \gStr) \cong \forall\kappa . \nats \times \forall\kappa . \later^\kappa \gStr$. To
prove that $\forall\kappa . \nats \cong \nats$ and $\forall\kappa . \later^\kappa \gStr \cong \Str$ ensuring the isomorphism
expected by a stream type, one needs two irrelevance axioms.

The first of these is the \emph{clock irrelevance axiom}
\begin{equation} \label{eq:cirr}
\inferrule*
{\hastype{\Gamma}{t}{\forall\kappa . A} \\ \kappa \notin \fc{A}}
{\hastype\Gamma{\cirr t}{\forall{\kappa'} . \forall{\kappa''} . \idty{A}{t [\kappa']}{t [\kappa'']}}}
\end{equation}
In the second hypothesis for the rule $\fc A$ stands for the free clocks of $A$ defined in the standard way. 
This rule can be used to prove that $\forall\kappa . A$ is isomorphic 
to $A$ if $\kappa$ is not free in $A$, in particular $\forall\kappa . \nats \cong \nats$. 
The second type isomorphism above requires the \emph{tick irrelevance axiom} 
\begin{equation}\label{eq:tirr}
\inferrule*
{\hastype{\Gamma}{t}{\later^\kappa A}}
{\tirr t :  \latbind{\tickA}{\kappa}  \latbind{\tickA'}{\kappa} \idty{A}{\tapp{t}}{\tapp[\tickA'] t}}
\end{equation}
which states that the identity of ticks is irrelevant for the equality theory, despite being crucial for the reduction semantics. 

Finally we mention the fixed point unfolding axiom \cite{GCTT}
\begin{equation} \label{eq:pfix}
\inferrule*
{\hastype[]{\Gamma}{t}{\later^\kappa A \to A} }
{\hastype{\Gamma}{\pfix[\kappa] \, t}{\latbind{\tickA}{\kappa}{\idty{A}{\tapp{(\dfix^{\kappa} t)}}{t(\dfix^{\kappa} t)}}}}
\end{equation}
In an extensional type theory this implies the judgemental fixed point unfolding equality 
$\tappc{(\dfix^\kappa \, t)} \jeq t\,(\dfix^\kappa\,t)$, and so, since the model presented in this paper is extensional, 
it will suffice to model $\pfix$. We write $\fix^\kappa t \defeq t(\dfix^\kappa t) : A$. Note that by extensionality, then 
\begin{align} \label{eq:fix}
t(\tabs\tickA\kappa{\fix^\kappa t})  & \jeq t(\tabs\tickA\kappa{\tapp{(\dfix^\kappa t)}}) 
 \jeq t(\dfix^\kappa t) \jeq \fix^\kappa t
\end{align}

Apart from the clock irrelevance axiom, the rules for universal quantification over 
clocks are exactly those for a $\Pi$-type indexed over $\clocktype$, except that $\clocktype$
is not a type. The latter means that $\clocktype$ can not appear positively in types, e.g.,
$\clocktype\to \clocktype$ is not wellformed.  To see why $\clocktype$ should not be a type,
note that clock irrelevance states that for a closed type $A$, all elements of $\forall\kappa . A$ are 
constant functions from clocks to $A$. Allowing $A = \clocktype$ would force all clocks to be equal.
In the model there will be an object $\clk$ modelling
$\clocktype$ and universal quantification over clocks will be modelled as a $\Pi$-type. 



\subsection{Universes}
\label{sec:syntax:universes}

In order to maintain consistency with the clock irrelevance axiom, universes in \clott\ are indexed by 
clock contexts. To see why this is necessary, note that 
naively adding a closed universe $\univ{}$, with a map $\later : \univ{} \to \forall\kappa . \univ{}$,
the clock irrelevance principle would state that the type operation $\later^\kappa (-)$ would
be independent of $\kappa$ on small types.
With the subscripting, the type operation $\later$ can be restricted on $\univ\Delta$ to the $\kappa \in \Delta$ avoiding this problem. 
Note that the formulation of \clott\ used here differs from that presented in \cite{bahr2017clocks}, which for simplicity used a single 
universe but retained consistency since the clock irrelevance axiom was not modelled (although mentioned in the paper). The presentation
of universes used here follows that of \gdtt~\cite{GDTT,GDTTmodel}, but extends it with ticks. 

The typing rules and equalities for universes are presented in Figure~\ref{fig:universes}. The subscript $\Delta$ of a universe is
a \emph{set} of clock variables, meaning in particular, that if $\Delta = \Delta'$ are equal as sets (contain the same elements),
then $\univ\Delta \jeq \univ{\Delta'}$. 
The universes are Tarski style, and we restrict
to a single universe level. The universes enjoy a form of polymorphism in the clock context: Inclusions of clock contexts induce
inclusions of universes, and these commute with the operations on the universe. For simplicity, we just include the rules for universal
quantification over clocks and $\later$. The rules for $\Pi$-, and $\Sigma$- types are the standard ones, indexed by a clock context,
plus a rule stating that these commute with the universe inclusions, see~\cite{GDTTmodel} for details. 
We also assume a code $\natscode$ for natural numbers in each universe.

\begin{figure}[tbp]
\textbf{Formation and typing rules}
\begin{mathpar}
 \inferrule*{\Delta = \{\kappa_1, \dots, \kappa_n\} \and \hastype{\Gamma}{\kappa_i}\clocktype \text{ for }i=1,\dots, n
 }
 {\istype{\Gamma}{\univ{\Delta}}} \and
 \inferrule*{
 \hastype{\Gamma}{t}{\univ{\Delta}}
 }
 {\istype{\Gamma}{\elems{\Delta}(t)}} 
 \and
 \inferrule*{
 \hastype{\Gamma}{t}{\univ{\Delta}} \\  \Delta \subseteq \Delta'
 }{
 \hastype{\Gamma}{\univin{\Delta}{\Delta'}(t)}{\univ{\Delta'}}}
 \and
  \inferrule*{
 \hastype{\Gamma,\kappa : \clocktype}A{\univ{\Delta,\kappa}}
 }
 {
  \hastype{\Gamma}{\forallcode{\Delta}\kappa . A}{\univ{\Delta}}
 } \and
 \inferrule*{
 \hastype{\Gamma, \tickA : \kappa}{A}{\univ{\Delta}} \and \kappa \in \Delta
 }{
 \hastype{\Gamma}{\latbindcode\tickA\kappa A}{\univ{\Delta}}
 }
\end{mathpar}
\textbf{Equations}
\begin{align*}
 \elems{\Delta}(\forallcode{\Delta'}\kappa . A) & \jeq \forall\kappa . \elems{{\Delta},\kappa} (A) &
 \elems{\Delta}{(\latbindcode \tickA \kappa A)} & \jeq \latbind\tickA\kappa \elems{\Delta}(A) \\
 \univin{{\Delta}}{\Delta'}{(\forallcode{\Delta'}\kappa . A)} & \jeq \forallcode{\Delta'}\kappa . \univin{({\Delta},\kappa)}{(\Delta',\kappa)}{(A)} &
 \univin{{\Delta}}{\Delta'}{(\latbindcode \tickA \kappa A)} & \jeq \latbindcode \tickA \kappa{\univin{\Delta}{\Delta'}{(A)}} \\
 \elems{\Delta'}(\univin{\Delta}{\Delta'} (t)) & \jeq \elems{\Delta}(t) &
 \univin{\Delta'}{\Delta''}(\univin{\Delta}{\Delta'}(t)) & \jeq \univin{\Delta}{\Delta''}(t) \\
 \univin{\Delta}{\Delta}(t) & \jeq t
\end{align*}
\begin{center}
\caption{Universes in \clott.}
\label{fig:universes}
\end{center}
\end{figure}

As mentioned above, guarded recursive types can be encoded as fixed points on the universe. For example,
$\gStr \defeq \elems\kappa(\gStrcode)$ where $\gStrcode \defeq \fix^\kappa (\lambda X. \natscode \timescode \latbindcode\tickA\kappa{\tapp X})$,
and $\timescode$ is the code for binary products encoded using $\Sigma$-types in the standard way. By (\ref{eq:fix}) this gives
\begin{align*}
 \gStr & \jeq \elems\kappa(\natscode \timescode \latbindcode\tickA\kappa{\gStrcode})\\
 & \jeq \nats \times \elems\kappa(\latbindcode\tickA\kappa{\gStrcode}) \\
 & \jeq \nats \times \latbind\tickA\kappa{\gStr} 
\end{align*}
Similarly, if $\syncode P : \nats \to \univ{\kappa}$ and $P(x) \jeq \elems\kappa{(\syncode P(x))}$ we can construct a lifting
$\predliftbare P$ of $P$ to a predicate on guarded streams as in Example~\ref{ex:streams}.
For this, define $\predlift P{xs} \defeq \elems\kappa(\predliftcode P(xs))$ where 
\[
 \predliftcode P \defeq \fix^\kappa(\lambda X . \lambda(\cons x{xs}) . \syncode P(x) \timescode 
 \latbindcode\tickA\kappa{\tapp X(\tapp{xs})}) : \gStr \to \univ\kappa
\]
Here the type of the variable $X$ is $\later^\kappa (\gStr \to \univ\kappa)$. As above, one can then 
verify that $\predlift P{\cons x{xs}} \jeq P(x) \times \latbind\tickA\kappa{\predlift P{\tapp{xs}}}$.

The presentation of \clott\ in~\cite{bahr2017clocks} 
had guarded recursive types as a primitive type formation rule.
This was because the version of \clott\ used there did not have identity types, and 
so did not have the fixed point unfolding axiom (\ref{eq:pfix}). Fixed points only 
unfolded when applied to $\tickc$. As a consequence (\ref{eq:fix}) did 
not hold, so the encoding of recursive types as fixed points on the universe
was not possible. Note that in an intensional version of \clott, equality (\ref{eq:fix}) 
holds only propositionally, making the guarded recursive types unfold only up 
to equivalence of types as in Guarded Cubical Type Theory~\cite{GCTT}.

The model constructed in this paper models extensional \clott\ including the axioms
(\ref{eq:cirr}), (\ref{eq:tirr}) and (\ref{eq:pfix}).

\section{A presheaf category}
\label{sec:basic:model}

The setting for the denotational semantics of \clott\ is a category of covariant presheaves over a category $\catT$ of time objects, 
which we now define. 
This category has previously been used to give a model of \gdtt~\cite{GDTTmodel} and 
a slight variant has been used to model Guarded Computational Type 
Theory~\cite{sterling2018guarded}.

We will assume given a countably infinite set $\cv$ of (semantic) clock variables, for which we use $\lambda, \lambda',\dots$ to range over. 
A \emph{time object} is a pair $\timeobj\FSA\vartheta$ where $\FSA$ is a finite subset of $\cv$ and 
$\vartheta:\FSA\rightarrow \nats$ is a map giving the number of ticks left on each clock in $\FSA$. We will write the finite sets
$\FSA$ as lists writing e.g., $\FSA, \lambda$ for $\FSA \cup\{\lambda\}$ and $\subex\vartheta\lambda n$ for the extension
of $\vartheta$ to $\FSA,\lambda$, or indeed for the update of $\vartheta$, if $\vartheta$ is already defined on $\lambda$. 
The time objects form a category $\catT$ whose morphisms $\timeobj\FSA\vartheta \to \timeobj{\FSA'}{\vartheta'}$ are functions 
$\tau:\FSA \rightarrow \FSA'$ such that 
$\vartheta'\tau\leq \vartheta$ in the pointwise order. The inequality allows for time to pass in a morphism, but morphisms can
also synchronise clocks in $\FSA$ by mapping them to the same clock in $\FSA'$, or introduce new clocks if $\tau$ is not
surjective. Define $\grtotal$ to be the category $\set^{\catT}$ of \emph{covariant} presheaves
on $\catT$. The topos of trees~\cite{Birkedal-et-al:topos-of-trees} 
can be seen as a restriction of this where time objects always have a single clock. 

If $\Gamma$ is a presheaf, $\gamma\in \Gamma\timeobj\FSA\vartheta$ and 
$\sigma : \timeobj\FSA\vartheta \to \timeobj{\FSA'}{\vartheta'}$, we will write $\sigma \cdot \gamma$ for $\Gamma(\sigma)(\gamma)$, 
the functorial action of $\Gamma$ applied to $\gamma$. With this notation, a presheaf is simply an indexed family of sets with actions
satisfying 
\begin{align}\label{eq:cat:action}
 \sigma\cdot (\tau \cdot \gamma) & = (\sigma\circ \tau)\cdot \gamma & \id \cdot \gamma & = \gamma
\end{align}
As for any presheaf category, $\grtotal$ caries a natural CwF structure in which a family over a presheaf $\Gamma$
is a presheaf over the category $\int \Gamma$ of elements of $\Gamma$. Recall that this has as objects pairs 
$\pair{\timeobj\FSA\vartheta}{\gamma}$ such that $\gamma \in \Gamma\timeobj\FSA\vartheta$, and morphisms
from $\pair{\timeobj\FSA\vartheta}{\gamma}$ to $\pair{\timeobj{\FSA'}{\vartheta'}}{\gamma'}$ morphisms 
$\sigma : \timeobj\FSA\vartheta \to \timeobj{\FSA'}{\vartheta'}$ of $\catT$ such that $\sigma \cdot \gamma = \gamma'$. 
Unfolding this definition, a family is a collection of sets $A_{\timeobj\FSA\vartheta}(\gamma)$ and maps
$\sigma : A_{\timeobj\FSA\vartheta}(\gamma) \to A_{\timeobj{\FSA'}{\vartheta'}}(\sigma \cdot\gamma)$ 
satisfying (\ref{eq:cat:action}). An element of $A$ is a family of elements 
$t_{\timeobj\FSA\vartheta}(\gamma) \in A_{\timeobj\FSA\vartheta}(\gamma)$ such that 
$\sigma\cdot (t_{\timeobj\FSA\vartheta}(\gamma)) = t_{\timeobj{\FSA'}{\vartheta'}}(\sigma \cdot \gamma)$. 
We will often omit the subscript and simply write $A(\gamma)$ and $t(\gamma)$. 
Abstractly, an element of $A$ is simply a global element of $A$ considered as a 
covariant presheaf over $\int\Gamma$. If $f : A \to B$ is a morphism of presheaves, and 
$t$ is an element of $A$, we can thus compose $f$ and $t$ to get an element $f \circ t$ of $B$.

Recall the following standard lemma~\cite{Hofmann1997}.

\begin{lem}
 The CwF $\grtotal$ models $\Pi$, $\Sigma$ and extensional identity types. 
\end{lem} 

Modelling $\clocktype$ as the object $\clk$ in $\grtotal$ defined as 
\begin{align*}
\clk\timeobj\FSA\vartheta & = \FSA & \tau\cdot \lambda & = \tau(\lambda)
\end{align*}
universal quantification can be modelled as a $\Pi$-type over $\clk$. 

\section{A dependent right adjoint}
\label{sec:dradjoint}

This section defines the dependent right adjoint to be used for modelling ticks in \clott. To talk about ticks we need a clock
in hand, and the smallest setting this happens in is the syntactic context $\kappa : \clocktype$, 
modelled as $\clk$.
In the CwF contexts extending this small context can be considered 
presheaves over the category $\int\clk$ of elements of $\clk$, 
and so in the following we will construct a dependent right adjoint on this. 
In Section~\ref{sec:modelling:ticks} we will see how to lift this to model ticks in \clott.

We write $\TStar$ for $\int \clk$ and $\grstar\defeq \set^{\TStar}$. Spelling out the definition, 
an object of $\TStar$ is a  triple $\triple\FSA\vartheta\lambda$
where $\lambda \in\FSA$ and a morphism $\triple\FSA\vartheta\lambda$ to $\triple{\FSA'}{\vartheta'}{\lambda'}$
is a morphism $\sigma : \timeobj\FSA\vartheta \to \timeobj{\FSA'}{\vartheta'}$ such that $\sigma(\lambda) = \lambda'$. 

\subsection{The right adjoint $\tlater$}

Recall first that in the topos of trees the functor $\tlater$ is defined as 
$(\tlater F)(n + 1) = Fn$ and $(\tlater F)(0) = \{\ast\}$. 
This generalises in a straightforward way to $\grstar$ by 
\begin{align*}
& (\tlater F) \triple\FSA\vartheta \lambda = 
\begin{cases} 
F \triple{\FSA}{\vartheta[\lambda-]}\lambda &\vartheta(\lambda)>0\\
\{\ast\}& \text{otherwise}
\end{cases}
\end{align*}
where $\vartheta[\lambda-] (\lambda)= \vartheta(\lambda) - 1$ and $\vartheta[\lambda-](\lambda') = \vartheta(\lambda')$ for $\lambda' \neq \lambda$.
The presheaf action of
$\tlater F$ is simply inherited from $F$ by noticing that a map $\sigma : \TobjA \to \TobjB$, 
induces a map
\[
\cminus\sigma\lambda : \TobjAminus \to \TobjBminus \, .
\]

\begin{lem} \label{lem:later:types:terms}
 The functor $\tlater : \grstar \to \grstar$ extends to families and elements.
\end{lem}
\begin{proof}
If $A$ is a type over $\Gamma$ and $\gamma \in (\tlater \Gamma) \triple \FSA \vartheta \lambda$  define
\[(\tlaterTy A)_{\triple \FSA \vartheta \lambda}(\gamma) = 
            \begin{cases} 
                      \{\ast\} & \vartheta(\lambda) = 0\\
                      A_{\triple \FSA {\vartheta[\lambda-]} \lambda}(\gamma) & \text{otherwise}
                      \end{cases}
                      \]
To see that this commutes with comprehension, note that 
${(\compr{\tlater\Gamma}{\tlaterTy A}){\triple \FSA \vartheta \lambda}}$ equals
$\tlater(\compr\Gamma A) {\triple \FSA \vartheta \lambda}$ when 
$\vartheta(\lambda) > 0$ and $\{\ast\}\times \{\ast\}$ when 
$\vartheta(\lambda) = 0$. The definition for elements is similar. 
%
\end{proof}

\begin{exa} 
As an example of a model of a type, recall the type of guarded streams satisfying 
$\gStr[\kappa] \jeqty \nats \times \later^\kappa \gStr[\kappa]$ from Section~\ref{sec:clott}.
This type is definable in the clock context $\kappa : \clocktype$, and so will be 
interpreted as a presheaf  in $\grstar$ defined as
$\pret{\Str^\kappa} \triple \FSA \vartheta \lambda =\nats^{\vartheta(\lambda)+1} \times \{\ast\}$.
We will assume that the products in this associate to the right, so that this is the type of tuples of the form
$(n_{\vartheta(\lambda)},(\dots, (n_0, \ast))\dots)$. This is needed to model the equality 
$\gStr[\kappa] \jeqty \nats \times \later^\kappa \gStr[\kappa]$, rather than just an isomorphism of types.
Given a predicate $x :  \Nat\vdash P$ the lifting of $P$ to streams $\predliftbare P$, 
described in Section~\ref{sec:syntax:universes} can be modelled as
\[
\pret{\predliftbare P}_{\triple \FSA \vartheta \lambda}(n_{\vartheta(\lambda)},(\dots, (n_0, \ast))\dots)
= \{(x_{\vartheta(\lambda)},(\dots, (x_0, \ast))\dots) \mid \forall i . x_i \in \pret P_{\triple \FSA \vartheta \lambda}(n_i) \}
\]
It is a simple calculation (using the definitions below) that these interpretations model the type equalities mentioned above. 
\end{exa}

\subsection{The left adjoint $\tearlier$}

In the topos of trees, the functor $\tlater$ defined above has a left adjoint $\tearlier$ defined as $(\tearlier F) n = F (n+1)$. 
At first sight it would seem that one can similarly define a left adjoint $\tearlier$ to $\tlater$ on $\grstar$ by
$(\tearlier F)\triple{\FSA}{\vartheta}{\lambda} = F\triple\FSA{\vartheta[\lambda +]}{\lambda}$, 
where $\vartheta[\lambda +]$ is defined similarly to
$\vartheta[\lambda-]$.
Unfortunately, $\tearlier F$ so described is not a presheaf because it has no well-defined action on maps since a map 
$\tau: \triple{\FSA}{\vartheta}\lambda \rightarrow \triple{\FSA'}{\vartheta'}{\lambda'}$ does not necessarily induce a map 
$\triple{\FSA}{\vartheta[\lambda+]}\lambda \to \triple{\FSA'}{\vartheta'[\lambda'+]}{\lambda'}$: If $\lambda'' \neq \lambda$
satisfies 
$\tau(\lambda'') = \lambda'$ there is no guarantee that 
$\vartheta'[\lambda'+](\tau(\lambda'')) \leq \vartheta[\lambda+](\lambda'')$.

To define the left adjoint, we instead first give an abstract description of $\tlater$.
Let $\TStarZ$ be the category defined as
$\TStar$, except that $\vartheta$ in an object $\triple\FSA\vartheta \lambda$ is a map of type $\FSA \to \ints$, i.e., 
the values can be negative. There is an inclusion $\phi : \TStar \to \TStarZ$ and we say that an object in 
$\TStarZ$ is \emph{negative} if it is not in the image of this inclusion. Note that if $\sigma : \triple\FSA\vartheta\lambda 
\to \triple{\FSA'}{\vartheta'}{\lambda'}$ and $\triple\FSA\vartheta\lambda$ is negative, so is $\triple{\FSA'}{\vartheta'}{\lambda'}$. 
Recall that $\phi$ induces a functor on presheaves $\phi^* : \Set^{\TStarZ} \to \Set^{\TStar}$ by 
$(\phi^*F)\triple\FSA\vartheta\lambda = F(\phi\triple\FSA\vartheta\lambda)$. 
The right adjoint
\[
  \phi_*: \Set^{\TStar} \to \Set^{\TStarZ}
\]
to $\phi^*$ can be defined as $\phi_*(F)\triple\FSA\vartheta\lambda = 1$ if $\triple\FSA\vartheta\lambda$ is negative and 
$F\triple\FSA\vartheta\lambda$ if not. There is a functor $[\star-] : \TStarZ \to \TStarZ$ mapping 
$\triple\FSA\vartheta\lambda$ to $\triple\FSA{\vartheta[\lambda-]}f$, where $\vartheta[\lambda-]$ is defined
as above. The functor
$\tlater$ can now be described as the composition
\[
\begin{diagram}
 \Set^{\TStar} \ar{r}{\phi_*} & \Set^{\TStarZ} \ar{r}{[\star-]^*} & \Set^{\TStarZ} \ar{r}{\phi^*} & \Set^{\TStar}
\end{diagram}
\] 
Since each of the functors in this diagram has a left adjoint, so does $\tlater$. The left adjoint $\tearlier$ 
is the composite (in the opposite order) of the left adjoints of the functors above, i.e.
\[
\begin{diagram}
 \Set^{\TStar} \ar{r}{\phi_!} & \Set^{\TStarZ} \ar{r}{[\star-]_!} & \Set^{\TStarZ} \ar{r}{\phi^*} & \Set^{\TStar}
\end{diagram}
\]
%
%
Unfolding the left Kan extensions used in the definitions of $\phi_!$ and $[\star-]_!$, the left adjoint can be described concretely
as follows. An element of $(\tearlier F)\triple \FSA\vartheta \lambda$ 
is an equivalence class of
a pair of a map
$\sigma : \triple {\FSA'}{\vartheta'}{\lambda'} \to \triple \FSA\vartheta \lambda$ such that 
$\vartheta'(\lambda') > \vartheta(\lambda)$ 
and an element $x \in F\triple {\FSA'}{\vartheta'}{\lambda'}$, 
up to the equivalence relation generated by 
\[
(\sigma\circ \tau, x) \sim (\sigma, \tau \cdot x)
\]
The presheaf action is defined by
$\tau \cdot [\pair \sigma x] = [\pair {\tau\sigma}x]$
and the functorial action of $\tearlier$ is defined as 
$\tearlier(f) [\pair \sigma x] = [\pair \sigma {f(x)}]$.
We note the following.

\begin{lem} \label{lem:earlier}
The functor $\tearlier : \grstar \to \grstar$ is left adjoint to $\tlater$.  
\end{lem}

Since $\tlater$ extends to families and elements (Lemma~\ref{lem:later:types:terms})
by Lemma~\ref{lem:bijectivecorresp}, 
$\tearlier$ has a dependent right adjoint $\tlater$ defined as 
$\tlater A = (\tlaterTy A)[\eta]$. There is, moreover, a projection
$\p_{\tearlier} : \tearlier \to \id$ mapping an element
$[\pair \sigma x]$ in $\tearlier F \triple\FSA\vartheta\lambda$ to $\sigma \cdot x$.
Thus, by Proposition~\ref{prop:tick:soundness}, 
$\tearlier$ and $\tlater$ provide a model for the tick calculus on $\grstar$.

We note the following, 
which is needed for the soundness of the tick irrelevance axiom (\ref{eq:tirr}).

\begin{lem}
\label{lem:tirr}
The projections $\p_{\tearlier}, \tearlier(\p_{\tearlier}) : 
\tearlier\tearlier \to \tearlier$ are equal. 
\end{lem}

\begin{proof}
  This follows from a simple calculation: 
  \[\p_{\tearlier}([\pair \sigma {[\pair \tau x]}]) = \sigma \cdot [\pair \tau x] = 
  [\pair {\sigma\tau} x]\] and 
  $\tearlier(\p_{\tearlier}) ([\pair \sigma {[\pair \tau x]}])
  = [\pair \sigma{\tau \cdot x}] = [\pair {\sigma\tau}{x}]$. 
\end{proof}

The natural transformation $\p_{\tearlier} : \tearlier F \to F$ corresponds 
to a natural transformation $\nxt = \tlater(\p_{\tearlier}) \circ \eta : F \to \tlater F$,
which can be described as follows

\begin{align} \label{eq:next}
& \nxt_{\triple\FSA\vartheta \lambda}(\gamma) = 
\begin{cases} 
\tickmap{}\cdot \gamma &\vartheta(\lambda)>0\\
\ast & \text{otherwise}
\end{cases}
\end{align}
where
\begin{equation}\label{eq:tickmap}
 \tickmap{} : \triple {\FSA}{\vartheta} \lambda \to 
 \triple {\FSA}{\vartheta[\lambda-]} \lambda
\end{equation}
is the map induced by the identity on $\FSA$. 

\section{Modelling ticks}
\label{sec:modelling:ticks}

We now show how to use the dependent right adjoint structure constructed in the previous section to model ticks in \clott. 
First note that $\clk$ induces a family $\clkfam[\Gamma]$ in any context $\Gamma$ defined
as $(\clkfam[\Gamma])_{\timeobj\FSA\vartheta} (\gamma) = \FSA$, such that $\clkfam[\Gamma'][\sigma] =\clkfam$
for any $\sigma : \Gamma \to \Gamma'$. The rule for extending contexts with ticks assumes a 
clock $\hastype{\Gamma}\kappa\clocktype$, which semantically corresponds to an element of $\clkfam[\Gamma]$. 
The collection of pairs $\pair\Gamma\clkel$ where $\Gamma$ is an object of $\grtotal$ and 
$\clkel$ is an element of $\clkfam$, extends to a category where a morphism from $\pair\Gamma\clkel$ to 
$\pair{\Gamma'}{\clkel'}$
is a morphism $\sigma : \Gamma \to \Gamma'$ in $\grtotal$, such that $\clkel = \clkel'[\sigma]$. This category is 
isomorphic to the slice category of $\grtotal$ over $\clk$, which in turn is wellknown to be equivalent to the category
$\grstar$. 
The equivalence maps a pair $\pair\Gamma\clkel$ to the presheaf defined as
\[
\Phi\pair\Gamma\clkel \triple\FSA\vartheta\lambda = \{\gamma\in \Gamma\timeobj\FSA\vartheta \mid \clkel(\gamma) = \lambda\}
\]
The opposite direction maps $F$ to 
\[
\Psi(F)\timeobj\FSA\vartheta = \coprod_{\lambda\in \FSA}F \triple\FSA\vartheta\lambda
\] 
and element given by the first projection. 

Note, moreover, that there is an isomorphism of categories of elements 
\begin{equation} \label{eq:cat:elements}
 \int \Phi\pair\Gamma\clkel \cong \int \Gamma
\end{equation}
so families and elements over $\Phi\pair\Gamma\clkel$ in $\grstar$ correspond bijectively 
to families and elements over $\Gamma$ in $\grtotal$. Since the isomorphism (\ref{eq:cat:elements}) 
is natural in $\pair\Gamma\clkel$ this bijective correspondence commutes with substitution. 
The dependent right adjoint constructed in the previous section therefore gives the following. 

\begin{thm} \label{thm:dradjoint:gr:total}
 The following structure exists on $\grtotal$.
\begin{enumerate}
\item An operation mapping an object $\Gamma$ and an element $\clkel$ of $\clk$ over $\Gamma$
to an object $\tearlier^\clkel\Gamma$.
\item An operation mapping $\Gamma, \clkel$ as above and $\sigma : \Gamma' \to \Gamma$ to
\[
\tearlier^\clkel(\sigma) : \tearlier^{\clkel[\sigma]}\Gamma' \to \tearlier^\clkel\Gamma
\]
which is functorial, in the sense that $\tearlier^\clkel(\id) = \id$ and 
$\tearlier^\clkel(\sigma\tau) = \tearlier^\clkel(\sigma) \circ \tearlier^{\clkel[\sigma]}(\tau)$.
\item A transformation, i.e., a family of maps $\p_{\tearlier^\clkel} : \tearlier^\clkel\Gamma \to \Gamma$, natural
in the sense that $\p_{\tearlier^\clkel} \circ \tearlier^\clkel\sigma 
= \sigma \circ \p_{\tearlier^{\clkel[\sigma]}}$.
\item \label{item:later:nat} An operation mapping families $A$ over $\tearlier^\clkel\Gamma$ to families $\tlater^\clkel A$
over $\Gamma$ satisfying ${(\tlater^\clkel A)[\sigma]} = \tlater^{\clkel[\sigma]}(A[\tearlier^\clkel(\sigma)])$.
\item\label{item:overline:nat} A bijection between elements of $A$ and elements of $\tlater^\clkel A$, mapping $t$ to
$\overline t$ natural in the sense that $\overline{t[\tearlier^\clkel \sigma]} = \overline t [\sigma]$.
\end{enumerate}
\end{thm}

In explicit terms, $\tearlier^\clkel\Gamma$ is defined as the first component of 
$\Psi(\tearlier\Phi\pair\Gamma\clkel)$. By (\ref{eq:cat:elements}), then 
\[
 \int \tearlier^\clkel\Gamma \cong  \int \Phi(\Psi(\tearlier(\Phi\pair\Gamma\clkel)) \cong \int \tearlier(\Phi\pair\Gamma\clkel)
\]
The operation on families in item (\ref{item:later:nat}) maps a family over $\tearlier^\clkel\Gamma$
to its corresponding family over $\tearlier(\Phi\pair\Gamma\clkel)$, then applies the dependent 
right adjoint operation on $\grstar$ to get a family over $\Phi\pair\Gamma\clkel$, and finally 
applies (\ref{eq:cat:elements}) to get a family over $\Gamma$. 

The structure of Theorem~\ref{thm:dradjoint:gr:total} is precisely what is required to model ticks and 
tick application in \clott. Note that the operation $\tlater$ extends to families and elements in such a 
way that 
\begin{equation} \label{eq:laterTy:eta}
 \tlater^\clkel A =  (\tlaterTy[\clkel] A)[\eta] 
\end{equation}
where $\eta : \Gamma \to  \tlater^{\clkel[\p_{\tearlier^\clkel}]} \tearlier^\clkel \Gamma$ is the unit of the adjunction. 

\section{Modelling guarded recursion}
\label{sec:guarded:rec}

Guarded fixed points can be modelled essentially as in~\cite{GDTTmodel}.
The aim of this section is to prove the following. 
\begin{lem} \label{lem:fix:points}
 For each ${\cwftm{\Gamma}{t}{\tlater^\clkel (A[\p_{\tearlier^\clkel}]) \to A}}$ there is a unique 
 $\cwftm{\Gamma}{\dfixe_{\Gamma, A}^\clkel(t)}{\tlater^\clkel (A[\p_{\tearlier^\clkel}])}$ satisfying 
 \[
 \dfixe_{\Gamma, A}^\clkel(t) = \transp{\cwfev t{\dfixe_{\Gamma, A}^\clkel(t)}[\p_{\tearlier^\clkel}]}
 \]
 Moreover, $\dfixe_{\Gamma, A}^\clkel(t)[\gamma] =  \dfixe_{\Gamma', A[\gamma]}^{\clkel[\gamma]}(t[\gamma])$
 for any $\gamma : \Gamma' \to \Gamma$. 
\end{lem}

The lemma will be proved by proving the corresponding lemma in $\grstar$. 
First we unfold the definitions of the families involved. 
The next lemma refers to the map $\tickmap{}$ of (\ref{eq:tickmap}). 

\begin{lem}
\begin{enumerate}
 \item  If $A$ is a family over $\Gamma$ in $\grstar$, 
 then the family ${\tlater (A[\p_{\tearlier}])}$ also over 
 $\Gamma$ can be described as
 \[
{\tlater (A[\p_{\tearlier}])}_{\triple{\FSA}\vartheta \lambda}(\gamma) = 
\begin{cases}
 \{\ast\} & \text{ if } \vartheta(\lambda) = 0 \\
 {A}_{\triple {\FSA}{\vartheta[\lambda-]} \lambda} (\tickmap{}\cdot \gamma) 
 & \text{ if } \vartheta(\lambda) = n+1
\end{cases}
\]
 \item If $A$ is a family over $\Gamma$ and $t$ is an element of $A$, then
\[
\transp{t[\p_{\tearlier^\clkel}]}_{\triple{\FSA}\vartheta \lambda}(\gamma) = 
\begin{cases}
 \ast & \text{ if } \vartheta(\lambda) = 0 \\
 t_{\triple {\FSA}{\vartheta[\lambda-]} \lambda} (\tickmap{}\cdot \gamma)
 & \text{ if } \vartheta(\lambda) = n+1
\end{cases}
\]
\end{enumerate}
\end{lem}


\begin{proofof}{Lemma~\ref{lem:fix:points}}
 We prove the corresponding statement in $\grstar$. 
 The proof is essentially the same as in~\cite{GDTTmodel}. Unfolding the equation of the lemma gives
\begin{align*}
 \dfixe_{\Gamma, A}^\clkel(t)_{\triple{\FSA}\vartheta \lambda}(\gamma) 
 & = \cwfev t{\dfixe_{\Gamma, A}^\clkel(t)}_{\triple {\FSA}{\vartheta[\lambda-]} \lambda} (\tickmap{}\cdot \gamma) \\
 & = t_{\triple {\FSA}{\vartheta[\lambda-]} \lambda} (\tickmap{}\cdot \gamma)(\id)(\dfixe_{\Gamma, A}^\clkel(t)_{\triple {\FSA}{\vartheta[\lambda-]} \lambda} (\tickmap{}\cdot \gamma))
\end{align*}
in the case where $\vartheta(\lambda)>0$ and $\ast$ else. Thus 
$\dfixe_{\Gamma, A}^\clkel(t)_{\triple{\FSA}\vartheta f}(\gamma)$ can be defined by induction on 
$\vartheta(\lambda)$. The last statement follows from uniqueness. 
\end{proofof}

%

\section{Modelling $\tickc$}
\label{sec:tickc}

To model $\tickc$ we will construct a substitution $d$ from the interpretation of any syntactic context of the form
$\Gamma, \kappa : \clocktype$ to the interpretation of $\Gamma, \kappa : \clocktype, \tickA : \kappa$, 
semantically substituting $\tickc$ for $\tickA$. Omitting denotation brackets for $\Gamma$, 
the type of $d$ is precisely
\[
 d : \compr\Gamma{\clkfam} \to \tearlier^\q(\compr\Gamma{\clkfam}) 
\]
Note that $\p_{\tearlier^\q}$ is a map in the opposite direction. We will show the following. 

\begin{prop} \label{prop:p:inverse}
 For any $\Gamma$ in  $\grtotal$, the map $\p_{\tearlier^\q} : 
 \tearlier^\q(\compr\Gamma{\clkfam}) \to \compr\Gamma{\clkfam}$ is an isomorphism. 
\end{prop}

\begin{proof}
 Under the isomorphism of Section~\ref{sec:modelling:ticks} the pair $(\compr\Gamma{\clkfam}, \q)$ corresponds to 
 the object $\Gamma_\star$ of $\grstar$ defined as $\Gamma_\star\triple\FSA\vartheta\lambda = \Gamma\timeobj\FSA\vartheta$.
 Recall that an element of $\tearlier(\Gamma_\star)\triple\FSA\vartheta\lambda$ is an equivalence class represented by a pair $\pair{\sigma}{x}$ where
 $\sigma : \triple{\FSA'}{\vartheta'}{\lambda'} \to \triple\FSA\vartheta\lambda$ such that $\vartheta'(\lambda') > \vartheta(\lambda)$
 and $x \in \Gamma_\star\triple{\FSA'}{\vartheta'}{\lambda'}$. The equivalence is the smallest equivalence relation relating
 $\pair{\sigma}{\tau\cdot x}$ to $\pair{\sigma\circ\tau}{x}$. Recall also that $\p_{\tearlier}([\pair{\sigma}{x}]) = \sigma \cdot x$.  
 
 To construct an inverse $d$ to $\p_{\tearlier}$, suppose now that $x \in \Gamma_\star\triple\FSA\vartheta\lambda$, and let 
 $n = \vartheta(\lambda)$. Let $\lambda'$ be fresh, and consider the mapping 
 $\iota : \timeobj\FSA\vartheta \to \timeobj{(\FSA, \lambda')}{\subex\vartheta{\lambda'}{m}}$ for some $m>n$ 
 given by the inclusion of $\FSA$ into $(\FSA, \lambda')$. Define 
 $d(x) = [\pair{\basicsub{\lambda'}\lambda}{\iota\cdot x}]$ where 
 \[
 \basicsub{\lambda'}\lambda : \triple{(\FSA, \lambda')}{\subex\vartheta{\lambda'}{m}}{\lambda'}
 \to \triple\FSA\vartheta\lambda
 \]
 maps $\lambda'$ to $\lambda$ and is the identity on all other input, and $\iota\cdot x$ refers to the action of $\iota$ on the
 object $\Gamma$. Clearly, $d$ is independent of the choices of $\lambda'$ and $m$. To show that $d$ is an inverse to $\p_{\tearlier}$, 
 note first that
\begin{align*}
 \p_{\tearlier}(d(x)) & = \p_{\tearlier}([\pair{\basicsub{\lambda'}\lambda}{\iota\cdot x}]) = 
 \basicsub{\lambda'}\lambda\cdot \iota \cdot x = x
\end{align*}
 For the other direction, suppose $[\pair\sigma x] \in \tearlier(\Gamma_\star)\triple\FSA\vartheta\lambda$, where 
 $\sigma : \triple{\FSA'}{\vartheta'}{\lambda'} \to \triple\FSA\vartheta\lambda$. Suppose $\lambda''$ fresh for both
 $\FSA$ and $\FSA'$, and let $m$ be a number strictly greater than both $\vartheta(\lambda)$ and $\vartheta'(\lambda')$.
 Consider the following commutative diagram in $\TStar$
 \[
\begin{diagram}
 \triple{(\FSA', \lambda'')}{\subex{\vartheta'}{\lambda''}{m}}{\lambda''} \ar{r}{\basicsub{\lambda''}{\lambda'}} 
 \ar{d}[swap]{\subex{\sigma}{\lambda''}{\lambda''}} & 
 \triple{\FSA'}{\vartheta'}{\lambda'} \ar{d}{\sigma} \\
 \triple{(\FSA, \lambda'')}{\subex{\vartheta}{\lambda''}{m}}{\lambda''} \ar{r}{\basicsub{\lambda''}{\lambda}} 
 & \triple{\FSA}{\vartheta}{\lambda} 
\end{diagram}
 \]
 and let $\iota' : \timeobj{\FSA'}{\vartheta'} \to \timeobj{(\FSA', \lambda'')}{\subex{\vartheta'}{\lambda''}{m}}$ be given by
 the inclusion. Then 
\begin{align*}
 [\pair\sigma x] & = \left[\pair\sigma{ \basicsub{\lambda''}{\lambda'} \cdot \iota' \cdot x}\right]\\ 
 & = \left[\pair{\sigma\circ \basicsub{\lambda''}{\lambda'}}{\iota'\cdot x}\right] \\
 & = \left[\pair{\basicsub{\lambda''}{\lambda} \circ \subex{\sigma}{\lambda''}{\lambda''}}{\iota'\cdot x}\right]  \\
 & = \left[\pair{\basicsub{\lambda''}{\lambda}}{\subex{\sigma}{\lambda''}{\lambda''} \cdot \iota'\cdot x}\right] \\
 & = \left[\pair{\basicsub{\lambda''}{\lambda}}{\iota\cdot \sigma \cdot x}\right]  \\ 
 & = d(\p_{\tearlier}(\left[\pair \sigma x\right])) \qedhere
\end{align*}
\end{proof}

\section{Universes}
\label{sec:semantic:universes}

The universes $\univ{\Delta}$ of \clott\ can be modelled by the semantic universes 
$\Usem{\Delta}$ constructed by Bizjak and M{\o}gelberg~\cite{GDTTmodel}. This section
recalls these and shows how to model the code $\latbindcode\tickA\kappa A$ which was not 
present in the language modelled by Bizjak and M{\o}gelberg. 

To model the universe $\univ{\Delta}$ we must have a $\Delta$-indexed set of semantic
clocks in hand. Semantically, this assumption can be represented by the object $\clk^\Delta$
of $\grtotal$ defined as $\clk^\Delta\timeobj\FSA\vartheta = \FSA^\Delta$, where the right 
hand side is to be understood as a set-theoretic exponent. We will therefore define $\Usem\Delta$
as a covariant presheaf over the category $\int \clk^\Delta$ 
of elements of $\clk^\Delta$. Recall that the latter
has as objects triples, $\triple\FSA\vartheta f$ where $f : \Delta \to \FSA$, and as morphisms 
$\triple\FSA\vartheta f$ to $\triple{\FSA'}{\vartheta'}{f'}$ morphisms 
$\sigma : \timeobj\FSA\vartheta \to \timeobj{\FSA'}{\vartheta'}$ such that $f' = \sigma \circ f$. 
We will write $\gr\Delta$ for the category of covariant presheaves on $\int \clk^\Delta$. Note that 
$\grtotal \cong \gr\emptyset$ and $\grstar \cong \gr{\{\kappa\}}$. 

As is standard, the construction of the universe $\Usem\Delta$ assumes a set-theoretic universe,
the inhabitants of which will be referred to as \emph{small} sets. This notion lifts to notions of
small families in $\grtotal$ and $\gr\Delta$ by requiring that all components be small.
The definition of the universe $\Usem\Delta$ refers to the notion of
invariance under clock introduction, whose definition we now recall. 

\begin{defi}\label{def:inv:clock:intro}
 A presheaf $F$ in $\gr\Delta$ is \emph{invariant under clock introduction} if, whenever
 $\lambda \notin \FSA$, the mapping $\iota\cdot(-) : F\triple\FSA\vartheta f
 \to F\triple{\FSA, \lambda}{\subex\vartheta\lambda n}{\iota f}$, induced by the inclusion
 of $\FSA$ into $\FSA, \lambda$, is an isomorphism. 
 A family $A$ over a presheaf $\Gamma$ is invariant under clock introduction if 
 the mapping $\iota\cdot(-) : A(\gamma)  \to A(\iota \cdot \gamma)$ is an isomorphism
 for each $\iota$ as above, and all $\gamma$. Note that $A$ can be 
 invariant under clock introduction also if $\Gamma$ is not. 
\end{defi}

For a presheaf $F$ in $\grtotal$ (considered as an object in $\gr{\emptyset}$) to be
invariant under clock introduction is essentially equivalent to the mapping $F \to F^{\clk}$
being an isomorphism~\cite{GDTTmodel}. 
So the condition captures the clock irrelevance axiom semantically. To model this axiom, 
it is therefore necessary that all types are interpreted as families invariant
under clock introduction. It is not necessary that contexts are invariant under clock introduction,
however, and they will not be, since the object $\clk$ is not. 
The standard Hofmann-Streicher universe~\cite{Hofmann-Streicher:lifting} 
in $\grtotal$ is not invariant under 
clock introduction, which is the semantic motivation for indexing universes by clock contexts,
see~\cite{GDTTmodel} for details. 

The universe $\Usem\Delta$ in $\gr\Delta$ is defined as follows. The component
$\Usem\Delta\triple\FSA\vartheta f$ is the set of small families over 
$y\timeobj{f[\Delta]}{\restrict{\vartheta}{f[\Delta]}}$ in $\grtotal$ invariant under clock introduction,
where $y$ is the yoneda embedding. 
In other words, an element of $\Usem\Delta\triple\FSA\vartheta f$ is a family of small sets 
$X_\tau$ indexed 
by morphisms $\tau$ of time objects with domain 
$\timeobj{f[\Delta]}{\restrict{\vartheta}{f[\Delta]}}$ where $f[\Delta]$ is the image of $f$, together
with maps $\sigma \cdot(-) : X_\tau \to X_{\sigma\tau}$ such that $\id \cdot x = x$ and 
$\sigma' \cdot (\sigma \cdot x) = (\sigma' \circ \sigma)\cdot x$ for all $x$, and such 
that $\iota \cdot(-)$ is an isomorphism, 
when $\iota : \timeobj{\FSA'}\vartheta \to \timeobj{\FSA', \lambda}{\subex\vartheta\lambda n}$
is given by the inclusion. The family of elements over $\Usem\Delta$ is defined as 
$\Elsem\Delta(X) = X_i$ where $i: \timeobj{f[\Delta]}{\restrict{\vartheta}{f[\Delta]}} 
\to \timeobj\FSA\vartheta$ is the inclusion. 

%
We now give a partial answer to the question of what the universes $\Usem{\Delta}$
classify. 

\begin{lem}[\cite{GDTTmodel}] \label{lem:univ:class}
 Let $A$ be a small family over an object $\Gamma$ in $\gr\Delta$. If both $\Gamma$
 and $A$ are invariant under clock introduction, there is a unique $\code A : \Gamma \to 
 \Usem\Delta$ such that $A = \Elsem\Delta[\code A]$. 
\end{lem}

The definition of the type operations on the universes rely on the property that all type operations
preserve invariance under clock introduction. 

\begin{lem}[\cite{GDTTmodel}] \label{lem:closure:iuci}
 The collection of families in $\gr\Delta$ invariant under clock introduction is 
 closed under reindexing and under forming $\Pi$- and $\Sigma$-types
 as well as identity types. The objects
 $\Usem{\Delta}$ and families $\Elsem{\Delta}$ are invariant under clock introduction.
\end{lem}

To model the universes in \clott, note that assumptions of the type formation rule for universes
semantically corresponds to a presheaf $\Gamma$ and a set $\clkset$ of elements of  
$\clkfam$. Suppose now that we are given some surjection $\Delta \to \clkset$
for some set $\Delta$. The reader could think of $\Delta$ as the syntactic set of clocks in
the syntactic universe $\univ\Delta$, and the surjection as the function mapping a 
clock to its interpretation, but in fact the choice of $\Delta$ does not matter. 
The surjection defines a map $\clkmap : \Gamma \to \clk^\Delta$, and we define 
$\Usem\clkset \defeq\Usem\Delta[\clkmap]$ as a family over $\Gamma$ 
and $\Elsem\clkset \defeq \Elsem{\Delta}[\cpair{\clkmap\circ \p}\q]$ as a family over 
$\compr{\Gamma}{\Usem\clkset}$. 

\begin{lem}[\cite{GDTTmodel}] \label{lem:univ:reindexing}
 The object $\Usem\clkset$ and the family $\Elsem\clkset$ are welldefined in the sense that 
 they are independent of the choice of $\Delta$ and surjection $\Delta \to\clkset$. Moreover,
 if $\rho : \Gamma' \to \Gamma$ then $\Usem{\clkset[\rho]} = \Usem\clkset[\rho]$
where $\{\kappa_1, \dots, \kappa_n\}[\rho] = \{\kappa_1[\rho], \dots, \kappa_n[\rho]\}$, for 
$\rho : \Gamma' \to \Gamma$ and likewise for $\Elsem{}$.
\end{lem}

The second of these statements follows from the first, since given a surjection $\Delta \to \clkset$
the composite $\Delta \to\clkset \to \clkset[\rho]$, where the second map maps $\kappa_i$ to 
$\kappa_i[\rho]$, is also a surjection inducing the map $\clkmap\circ\rho : \Gamma' \to \clk^\Delta$.
Therefore $\Usem{\clkset[\rho]} = \Usem{\Delta}[\clkmap\circ\rho] = \Usem\clkset[\rho]$. 

In the next section we show how to model 
also $\latbindcode\tickA\kappa (-)$, but first we recall the construction for universe inclusion. 

Suppose $\clkset \subseteq \clkset'$ and suppose that we are given a surjection $\Delta' \to \clkset'$. Let 
$\Delta\subseteq \Delta'$ be the preimage of $\clkset$. There is a projection 
$p_{\Delta, \Delta'} : \clk^{\Delta'} \to \clk^\Delta$ and by Lemma~\ref{lem:closure:iuci} the object $\Usem\Delta[p_{\Delta, \Delta'}]$ 
and the family $\Elsem\Delta[\cpair{p_{\Delta, \Delta'}\circ \p}{\q}]$ are invariant under clock introduction. 
Therefore, by Lemma~\ref{lem:univ:class}, 
there is a unique map $\Usemin{}{\Delta}{\Delta'} : \Usem\Delta[p_{\Delta, \Delta'}] \to \Usem{\Delta'}$ 
in $\gr{\Delta'}$ such that 
$\Elsem\Delta[\cpair{p_{\Delta, \Delta'}\circ \p}{\q}] = \Elsem{\Delta'}[\Usemin{}{\Delta}{\Delta'}]$. Now let
$\clkmapB : \Gamma \to \clk^{\Delta'}$ be the map induced by the surjection from $\Delta'$ to $\clkset'$, 
and $\clkmap$ the one induced by the surjection $\Delta \to \clkset$. Then
$\clkmap \defeq p_{\Delta, \Delta'} \circ\clkmapB : \Gamma \to \clk^{\Delta}$.
Reindexing $\Usemin{}{\Delta}{\Delta'}$ along $\clkmapB$ gives a map 
$\Usemin{}{\clkset}{\clkset'} : \Usem{\clkset} \to \Usem{\clkset'}$ in the category of covariant presheaves over 
$\int \Gamma$ such that 
$\Elsem{\clkset'}[\Usemin{}{\clkset}{\clkset'}] = \Elsem{\clkset}$. Moreover, this map
can be proved~\cite{GDTTmodel} independent of the choice of $\Delta'$ and the surjection used to define $\clkmapB$. 
If $A$ is a family over $\Gamma$ and 
$\code A$ is an element of $\Usem{\clkset}$ such that $A = \Elsem{\clkset}[\cpair{\id_\Gamma}{\code A}]$, then
$\Usemin{}{\clkset}{\clkset'} \circ \code A$ is an element of $\Usem{\clkset'}$ and 
$\Elsem{\clkset'}[\cpair{\id_\Gamma}{\Usemin{}{\clkset}{\clkset'}\circ\code A}] = \Elsem{\clkset}[\cpair{\id_\Gamma}{\code A}] = A$.
Universe inclusion can thus be modelled by postcomposition by $\Usemin{}{\clkset}{\clkset'}$.

\subsection{Modelling $\latbindcode\tickA\kappa (-)$}

For this section let $\Delta = \{\kappa_1, \dots, \kappa_n\}$. 
Note first that the definition of the operation $\tlater$ on 
$\grstar \cong{\gr{\{\kappa\}}}$ can be extended to define an operation on $\gr\Delta$ 
for every $\kappa_i$:
\begin{align} \label{eq:later:kappa:i}
& (\tlater^{\kappa_i} A) \triple\FSA\vartheta f = 
\begin{cases} 
A \triple{\FSA}{\vartheta[f(\kappa_i)-]}f &\vartheta(f(\kappa_i))>0\\
\{\ast\}& \text{otherwise}
\end{cases}
\end{align}
and that this extends to families and elements by similar constructions as for the case of $\grstar$.
The functor $\tlater^{\kappa_i}$ actually has a left adjoint defined similarly to $\tearlier$, but we shall 
not need that here.

One way of viewing definition (\ref{eq:later:kappa:i}) is that $\kappa_i$ defines an element 
$\hat{\kappa_i}$ of the family $\clkfam[]$ over $\clk^\Delta$, and 
\begin{equation} \label{eq:later:kappa:i:2}
(\tlater^{\kappa_i} A) = (\tlaterTy[\hat{\kappa_i}] A)[\nxt^{\hat{\kappa_i}}]
\end{equation}
where $\nxt^{\hat{\kappa_i}} : \clk^\Delta \to \tlater^{\hat{\kappa_i}} (\clk^\Delta)$ is the map
of (\ref{eq:next}).
Suppose now that $\clkmap : \Gamma \to \clk^\Delta$ is
induced by some surjection $\Delta \to \clkset$ and let $\clkset_i$ be the image of $\kappa_i$ under this
surjection. If $A$ is an object in $\gr \Delta$, then by
(\ref{eq:later:kappa:i:2})
\begin{equation} \label{eq:later:kappa:i:reindex}
(\tlater^{\kappa_i} A)[\clkmap] = \tlaterTy[\clkset_i](A[\clkmap])[\nxt^{\clkset_i}]
\end{equation}

Suppose now that $B$ is a family over $A$ in $\gr\Delta$. Then $\tlaterTy[\kappa_i]B$ 
defines a family over $\compr{\clk^{\Delta}}{\tlater^{\kappa_i}A}$ in $\grtotal$, and 
direct calculation verifies that
\begin{equation} \label{eq:later:kappa:i:Fam}
  \tlaterTy[\kappa_i]B = 
  (\tlaterTy[\hat{\kappa_i}{[\p]}]B)
  [\inv{\cpair{\tlater^{\hat{\kappa_i}}(\p)}{\tlaterTm[\hat{\kappa_i}](\q)}}\circ \cpair{\nxt \circ \p}{\q}]
\end{equation}
Here the reindexing on the right is along the composite map
\[
\begin{diagram}
 \compr{\clk^{\Delta}}{\tlater^{\kappa_i}A} 
 \ar{r}{\cpair{\nxt \circ \p}{\q}} & \compr{\tlater^{\hat{\kappa_i}}\clk^{\Delta}}{\tlaterTy[\hat{\kappa_i}]A} 
 \ar{rr}{\inv{\cpair{\tlater^{\hat{\kappa_i}}(\p)}{\tlaterTm[\hat{\kappa_i}](\q)}}} 
 && \tlater^{\hat{\kappa_i}[\p]}(\compr{\clk^{\Delta}}{A})
\end{diagram}
\]
the first of which is well-typed by (\ref{eq:later:kappa:i:2}). If $\clkset$ and $\clkset_i$ are as above
then 
substituting both sides of (\ref{eq:later:kappa:i:Fam}) along $\cpair{\clkmap\circ \p}{\q}$  gives 
\begin{equation}\label{eq:later:kappa:i:Fam:reindex}
  (\tlaterTy[\kappa_i]B)[\cpair{\clkmap\circ \p}{\q}] = 
  (\tlaterTy[\clkset_i{[\p]}]B[\cpair{\clkmap\circ \p}{\q}])
  [\inv{\cpair{\tlater^{\clkset_i}(\p)}{\tlaterTm[\clkset_i](\q)}}\circ \cpair{\nxt \circ \p}{\q}]
\end{equation}

The next theorem gives the semantic structure needed to model $\latbindcode\tickA\kappa (-)$. 

\begin{thm}\label{thm:tlater:code}
 Let $\Gamma$ be an object of $\grtotal$, let $\clkset$ be a set of elements of  $\clkfam$, and let 
 $\clkset_i \in \clkset$. 
 There is a morphism 
 $\code{\tlater^{\clkset_i}_\clkset}: \tlater^{\clkset_i}(\Usem{\clkset[\p_{\tearlier^{\clkset_i}}]}) \to \Usem\clkset$ 
 in the category of presheaves over the elements of $\Gamma$ such that the following hold
 \begin{enumerate}
   \item \label{item:tlater:code:1} If $A$ is a family over $\tearlier^{\clkset_i}\Gamma$ and 
   $\code A$ is an element of 
   $\Usem{\clkset[\p_{\tearlier^{\clkset_i}}]}$ such that 
 \[\Elsem{\clkset[\p_{\tearlier^{\clkset_i}}]}[\cpair{\id_{\tearlier^{\clkset_i}\Gamma}}{\code A}] = A,\] 
 then
 $\Elsem{\clkset}[\cpair{\id_{\Gamma}}{\code{\tlater^{\clkset_i}_\clkset} \circ \overline{\code A}}] 
 = \tlater^{\clkset_i} A$. 
 \item \label{item:tlater:code:2} If $\sigma : \Gamma' \to \Gamma$ then 
 $\code{\tlater^{\clkset_i}_\clkset}[\sigma] = \code{\tlater^{\clkset_i[\sigma]}_{\clkset[\sigma]}}$.
 \item \label{item:tlater:code:3} If  $\clkset\subseteq \clkset'$ then 
 \[
   \Usemin{}{\clkset}{\clkset'} \circ \code{\tlater^{\clkset_i}_\clkset} \circ \overline{\code A}
   = \code{\tlater^{\clkset_i}_{\clkset'}}\circ 
   \overline{(\Usemin{}{(\clkset[\p_{\tlater^{\clkset_i}}])}{(\clkset'[\p_{\tlater^{\clkset_i}}])} \circ \code A)}
 \]
 \end{enumerate}
\end{thm}

The first of these items states that if $\code A$ is a code for $A$, then 
$\code{\tlater^{\clkset_i}_\clkset} \circ \overline{\code A}$ is a code for 
$\tlater^{\clkset_i} A$. The second one states that $\code{\tlater^{\clkset_i}_\clkset}$
is preserved by substitutions and the third that it commutes with universe inclusions. 
The second of these will be used in the proof of the substitution lemma 
(Lemma~\ref{lem:substitution}) and the other two in the proof of soundness 
(Theorem~\ref{thm:soundness}).

\begin{proof}
 First note that if $\Gamma$ is an object of $\gr\Delta$ and $A$ is a small family over $\Gamma$
 such that both $\Gamma$ and $A$ are invariant under clock introduction, 
 then also $\tlater^{\kappa_i} \Gamma$
 and the small family $\tlaterTy[\kappa_i] A$
 over $\tlater^{\kappa_i} \Gamma$ are invariant under clock introduction. This is simply
 because the action of an inclusion $\iota : \triple\FSA\vartheta f
 \to\triple{\FSA, \lambda}{\subex\vartheta\lambda n}{\iota f}$ on $\tlater^{\kappa_i} \Gamma$ 
 and $\tlaterTy[\kappa_i] A$
 is given by the action of $\iota : \triple{\FSA}{\vartheta[f({\kappa_i})-]}f
 \to \triple{\FSA, \lambda}{\subex{\vartheta[f({\kappa_i})-]}\lambda n}{\iota f}$ on $\Gamma$ and $A$, in
 the case of $f({\kappa_i}) > 0$ and by the mapping $\{\ast\} \to \{\ast\}$ when $f({\kappa_i}) = 0$. 
 In particular $\tlater^{{\kappa_i}}(\Usem\Delta)$ and 
 $\tlaterTy[{\kappa_i}](\Elsem\Delta)$ are both invariant under clock introduction, and therefore
 there is a unique morphism $\code{\tlater^{\kappa_i}_\Delta}: \tlater^{{\kappa_i}}(\Usem\Delta) \to \Usem\Delta$
 such that $\tlaterTy[{\kappa_i}](\Elsem\Delta) = \Elsem\Delta[\code{\tlater_\Delta^{{\kappa_i}}}]$. 
 Define $\code{\tlater^{\clkset_i}_\clkset} \defeq \code{\tlater^{{\kappa_i}}_\Delta}[\clkmap]$, which
 can be proved independent of the choice of $\Delta$ and surjection inducing $\clkmap$ using the
 tools of~\cite{GDTTmodel}. To see that this
 has the right domain, note that by (\ref{eq:later:kappa:i:reindex})
\begin{align*}
 (\tlater^{{\kappa_i}}(\Usem\Delta))[\clkmap] & = (\tlaterTy[\clkset_i]\Usem\clkset)[\nxt^{\clkset_i}] \\
 & = (\tlaterTy[\clkset_i]\Usem\clkset)[\tlater^{\clkset_i}(\p_{\clkset_i}) \circ \eta] \\
 & = (\tlaterTy[\clkset_i]\Usem{\clkset[\p_{\clkset_i}]})[\eta] \\
 & = \tlater^{\clkset_i}\Usem{\clkset[\p_{\clkset_i}]}
\end{align*}
using (\ref{eq:laterTy:eta}) in the last line.
Item (\ref{item:tlater:code:1}) can then be proved as follows. First note that
\begin{align*}
 \Elsem{\clkset}[\cpair{\id_{\Gamma}}{\code{\tlater^{\clkset_i}_\clkset} \circ \overline{\code A}}
 & = \Elsem{\Delta}[\cpair{\clkmap\circ \p}{\q}\circ 
 \cpair{\id_{\Gamma}}{\code{\tlater^{\clkset_i}_\clkset}\circ \q}
 \circ \cpair{\id_{\Gamma}}{\overline{\code A}}] \\
 & = \Elsem{\Delta}[\cpair{\clkmap\circ \p}{\q}\circ 
 \cpair{\id_{\Gamma}}{\code{\tlater^{\kappa_i}_\Delta}[\clkmap]\circ \q}
 \circ \cpair{\id_{\Gamma}}{\overline{\code A}}] \\
 & = \Elsem{\Delta}[\cpair{\id_{\Gamma}}{\code{\tlater^{\kappa_i}_\Delta}\circ \q} \circ
 \cpair{\clkmap\circ \p}{\q} \circ \cpair{\id_{\Gamma}}{\overline{\code A}}] \\
 & = \tlaterTy[\kappa_i](\Elsem\Delta)[\cpair{\clkmap\circ \p}{\q} \circ \cpair{\id_{\Gamma}}{\overline{\code A}}] 
\end{align*}
 By (\ref{eq:later:kappa:i:Fam:reindex})
\begin{align*}
 \tlaterTy[\kappa_i](\Elsem\Delta)[\cpair{\clkmap\circ \p}{\q}] & = 
 (\tlaterTy[\clkset_i{[\p]}]\Elsem\Delta[\cpair{\clkmap\circ \p}{\q}]) [\inv{\cpair{\tlater^{\clkset_i}(\p)}{\tlaterTm[\clkset_i](\q)}}\circ \cpair{\nxt \circ \p}{\q}] \\
 & = (\tlaterTy[\clkset_i{[\p]}]\Elsem\clkset) [\inv{\cpair{\tlater^{\clkset_i}(\p)}{\tlaterTm[\clkset_i](\q)}}\circ \cpair{\nxt \circ \p}{\q}]
\end{align*}
and so 
\begin{align*}
 \Elsem{\clkset}[\cpair{\id_{\Gamma}}{\code{\tlater^{\clkset_i}_\clkset} \circ \overline{\code A}}
 & = (\tlaterTy[\clkset_i{[\p]}]\Elsem\clkset)
  [\inv{\cpair{\tlater^{\clkset_i}(\p)}{\tlaterTm[\clkset_i](\q)}}\circ \cpair{\nxt \circ \p}{\q} \circ
  \cpair{\id_{\Gamma}}{\overline{\code A}}] \\ 
 & = (\tlaterTy[\clkset_i{[\p]}]\Elsem\clkset)
  [\inv{\cpair{\tlater^{\clkset_i}(\p)}{\tlaterTm[\clkset_i](\q)}}\circ \cpair{\nxt}{\overline{\code A}}]  
\end{align*}
Since $\nxt = \tlater^{\clkset_i}(\p_{\tearlier^{\clkset_i}}) \circ \eta$ and $\overline{\code A} = \tlaterTm[\clkset_i](\code A)[\eta]$
the above reduces to 
\begin{align*}
 \Elsem{\clkset}[\cpair{\id_{\Gamma}}{\code{\tlater^{\clkset_i}_\clkset} \circ \overline{\code A}}
 & = (\tlaterTy[\clkset_i{[\p]}]\Elsem\clkset)
  [\inv{\cpair{\tlater^{\clkset_i}(\p)}{\tlaterTm[\clkset_i](\q)}}\circ 
  \cpair{\tlater^{\clkset_i}(\p_{\tearlier^{\clkset_i}})}{\tlaterTm[\clkset_i](\code A)}\circ \eta]   \\
 & = (\tlaterTy[\clkset_i{[\p]}]\Elsem\clkset) [\tlater^{\clkset_i}(\cpair{\p_{\tearlier^{\clkset_i}}}{\code A})\circ \eta]  \\
 & = (\tlaterTy[\clkset_i](\Elsem\clkset[\cpair{\p_{\tearlier^{\clkset_i}}}{\code A}]))[\eta]  \\
 & = (\tlaterTy[\clkset_i](\Elsem{\clkset[\p_{\tearlier^{\clkset_i}}]}[\cpair{\id}{\code A}]))[\eta]  \\
 & = (\tlaterTy[\clkset_i](A))[\eta]  \\
 & = \tlater^{\clkset_i}(A) 
\end{align*}

For item (\ref{item:tlater:code:2}) note that given a surjection $\Delta \to \clkset$, the composite map
$\Delta \to \clkset \to \clkset[\sigma]$ where the last map maps $\kappa$ to $\kappa[\sigma]$ is also 
a surjection inducing the map $\clkmap\circ \sigma : \Gamma' \to \clk^\Delta$.
So  $\code{\tlater^{\clkset_i[\sigma]}_{\clkset[\sigma]}}$ can be defined as 
$\code{\tlater^{{\kappa_i}}_\Delta}[\clkmap\circ \sigma]$ which equals $\code{\tlater^{\clkset_i}_\clkset}[\sigma]$. 

For item (\ref{item:tlater:code:3}), first note that by (\ref{eq:later:kappa:i:2})
\begin{align*}
 (\tlater^{\kappa_i}\Usem\Delta)[p_{\Delta, \Delta'}] & = (\tlaterTy[\hat{\kappa_i}]\Usem{p_{\Delta, \Delta'}})[\nxt^{\hat{\kappa_i}}]
  = \tlater^{\kappa_i}(\Usem{p_{\Delta, \Delta'}})
\end{align*}
The diagram
\[
\begin{diagram}
 \tlater^{\kappa_i}(\Usem{p_{\Delta, \Delta'}})  \ar{d}[swap]{\tlater^{\kappa_i}(\Usemin{}{\Delta}{\Delta'})} 
 \ar{rr}{\code{\tlater^{{\kappa_i}}_\Delta}[p_{\Delta, \Delta'}]} & &
  \Usem{p_{\Delta, \Delta'}} \ar{d}{\Usemin{}{\Delta}{\Delta'}}  \\
 \tlater^{\kappa_i}(\Usem{\Delta'}) \ar{rr}{\code{\tlater^{{\kappa_i}}_{\Delta'}}} && \Usem{\Delta'} 
\end{diagram}
\]
commutes by the uniqueness statement of Lemma~\ref{lem:univ:class} because both directions classify the family 
$\tlaterTy[\kappa_i](\Elsem{p_{\Delta, \Delta'}})$. Using this,  item (\ref{item:tlater:code:3}) follows.
\end{proof}

\section{Interpretation of syntax}
\label{sec:interp:syntax}

This section defines the interpretation of the syntax into the model. It is well known that to well define 
an interpretation of judgements (rather than derivations) of dependent type theory
with the conversion rule, the syntax of the type theory must be 
annotated with typing information. For example, Hofmann~\cite{Hofmann1997} annotates the application term for $\Pi$
types with the type of the function being applied, writing $\PiApp xABtu$ rather than the more common
application term $t\,u$, in order to define an interpretation of syntax into a general CwF. Likewise lambda
expressions are annotated not just with the type of the variable being abstracted, but also with the target type
of the abstraction. 

Following Hofmann~\cite{Hofmann1997} we define an annotated syntax of pre-terms, pre-types and pre-contexts 
(allowing also judgements with no derivations), define an interpretation of this into the model
as a partial function by structural induction, and finally prove that the interpretation of judgements
with derivations are well-defined. 

\begin{figure*}[tb]
\begin{center}
\textbf{Typing rules}
\begin{mathpar}
  \inferrule*
  {\hastype[]{\Gamma,\kappa : \clocktype}{t}{A}\\
    \wfcxt{\Gamma}}
  {\hastype{\Gamma}{\ForallLam\kappa At}{\forall \kappa . A}}
  \and
  \inferrule*
  {\hastype{\Gamma}{t}{\forall \kappa . A}\\
    \hastype\Gamma{\kappa'}\clocktype}
  {\hastype{\Gamma}{\ForallApp \kappa At{\kappa'}}{A \subst{\kappa}{\kappa'}}}
  \and
  \inferrule*
  {\hastype{\Gamma,\tickA:\kappa}{t}{A}}
  {\hastype{\Gamma}{\TickLam{\tickA}{\kappa}A t}{\latbind{\tickA}{\kappa} A}}
  \and
  \inferrule*
  {\hastype{\Gamma}{t}{\latbind{\tickA}{\kappa} A}\\\wfcxt{\Gamma,\tickB:\kappa,\Gamma'}}
  {\hastype{\Gamma,\tickB: \kappa,\Gamma'}{\TickApp\tickA\kappa A t{\tickB}}{A\toksubst{\tickB}{\tickA}}}
  \and
  \inferrule*
  {\hastype[]{\Gamma, \kappa : \clocktype}{t}{\latbind{\tickA}{\kappa} A}\\ \hastype{\Gamma}{\kappa'}\clocktype}
  {\hastype{\Gamma}{\TickcApp\tickA \kappa At{\kappa'}}{A\tickcsub\tickA\kappa{\kappa'}}} 
  \and
  \inferrule*
  {\hastype{\Gamma}{t}{\later^\kappa A \to A}}
  {\hastype{\Gamma}{\dfix_A^\kappa\,t}{\later^\kappa A}}
\end{mathpar}
\textbf{Tick constant substitution}
\begin{align*}
 (\ForallLam{\kappa''} At)\tickcsub\tickA\kappa{\kappa'} 
 & = \ForallLam{\kappa''} {A\tickcsubsm\tickA\kappa{\kappa'}}{(t\tickcsub\tickA\kappa{\kappa'})} \\ 
 \ForallApp {\kappa''} At{\kappa} \tickcsub\tickA\kappa{\kappa'} 
 & =  \ForallApp{\kappa''} {A\tickcsubsm\tickA\kappa{\kappa'}}{t\tickcsub\tickA\kappa{\kappa'}}{\kappa'} \\
 \ForallApp {\kappa''} At{\kappa'''} \tickcsub\tickA\kappa{\kappa'} 
 & =  \ForallApp{\kappa''} {A\tickcsubsm\tickA\kappa{\kappa'}}{t\tickcsub\tickA\kappa{\kappa'}}{\kappa'''} \\
 (\TickLam{\tickB}{\kappa}A t)\tickcsub\tickA\kappa{\kappa'} 
 & = \TickLam{\tickB}{\kappa'}{A\tickcsubsm\tickA\kappa{\kappa'}}{(t\tickcsub\tickA\kappa{\kappa'})} \\
 (\TickLam{\tickB}{\kappa''}A t)\tickcsub\tickA\kappa{\kappa'} 
 & = \TickLam{\tickB}{\kappa''}{A\tickcsubsm\tickA\kappa{\kappa'}}{(t\tickcsub\tickA\kappa{\kappa'})} \\
 \TickApp\tickB\kappa A t{\tickA}\tickcsub\tickA\kappa{\kappa'} 
 & =  \TickcApp\tickB\kappa At{\kappa'} \\
  \TickApp\tickB\kappa A t{\tickA''}\tickcsub\tickA\kappa{\kappa'} 
 & =  \TickApp\tickB{\kappa'} {A\tickcsubsm\tickA\kappa{\kappa'}}{t\tickcsub\tickA\kappa{\kappa'}}{\tickA''}\\
  \TickApp\tickB{\kappa''} A t{\tickA''}\tickcsub\tickA\kappa{\kappa'} 
 & =  \TickApp\tickB{\kappa''} {A\tickcsubsm\tickA\kappa{\kappa'}}{t\tickcsub\tickA\kappa{\kappa'}}{\tickA''} \\
 \TickcApp\tickB{\kappa''}At{\kappa}\tickcsub\tickA\kappa{\kappa'}
 & =  \TickcApp\tickB{\kappa''} {A\tickcsubsm\tickA\kappa{\kappa'}}{t\tickcsub\tickA\kappa{\kappa'}}{\kappa'} \\
 \TickcApp\tickB{\kappa''}At{\kappa'''}\tickcsub\tickA\kappa{\kappa'}
 & =  \TickcApp\tickB{\kappa''} {A\tickcsubsm\tickA\kappa{\kappa'}}{t\tickcsub\tickA\kappa{\kappa'}}{\kappa'''}
\end{align*}
\caption{Annotated typing rules for Clocked Type Theory.}
\label{fig:clott:typing:annotated}
\end{center}
\end{figure*}

The typing rules for annotated terms are presented in Figure~\ref{fig:clott:typing:annotated}. We 
omit the standard cases such as $\Pi$- and $\Sigma$-types, as well as universes. Operations 
on universes should be annotated with the context index of the universe 
at which they are applied, for example $\latbindcodeAnn{\Delta}\tickA\kappa A$ contains the clocks context $\Delta$.
The annotated
syntax is a straight-forward adaptation of the annotations used by Hofmann, except in the case of application
to $\tickc$. The first thing to note is that the substitution of $\kappa'$ for $\kappa$ in the conclusion of the 
typing rule for $\tickc$ has been replaced by an explicit substitution: $\TickcApp\tickA \kappa At{\kappa'}$
binds $\kappa$ in $t$ and applies it to $\kappa'$. This is not a completely faithful representation of the original
syntax of Clocked Type Theory: In the original syntax, one might have different terms $s$ and $t $ such that
$t\subst\kappa{\kappa'} = s\subst\kappa{\kappa'}$ in which case the terms $\tappc {t\subst\kappa{\kappa'}}$
and $\tappc {s\subst\kappa{\kappa'}}$ are syntactically equal, but $\TickcApp\tickA \kappa At{\kappa'}$
and $\TickcApp\tickA \kappa As{\kappa'}$ are not. This is unfortunate, but seems unavoidable. See 
Section~\ref{sec:conclusion} for a discussion.

Note also that $\TickcApp\tickA \kappa At{\kappa'}$ binds both $\tickA$ and $\kappa$ in $A$
(as indicated by the two sets of square brackets), unlike, e.g., $\TickLam{\tickA}{\kappa}A t$ which only binds
$\tickA$. The typing in the conclusion uses a special simultaneous substitution of $\tickc$ for $\tickA$ 
and $\kappa'$ for $\kappa$. This is defined for terms in the figure; on types this simply distributes 
over the structure of types, except in universes, where it substitutes $\kappa'$ for $\kappa$. 
The rules defining the substitution are subject to the usual side conditions avoiding capture of bound variables, but
we omit these from the figure. In case multiple rules match, the top-most one should be applied. For example, the third
rule only triggers when $\kappa \neq \kappa'''$.

The rules for judgemental equality can be formulated in the annotated syntax essentially as they are in the unannotated syntax. 
For example, the $\beta$ rule for ticks in the case of application to $\tickc$ is
\begin{equation} \label{eq:annotated:beta}
  \TickcApp\tickA \kappa A{(\TickLam{\tickA}{\kappa}A t)}{\kappa'} = t\tickcsub\tickA\kappa{\kappa'}
\end{equation}

\begin{figure*}[tbp]
\begin{center}
\textbf{Contexts}
\begin{align*}
 \den{\wfcxt{-}} & = 1 &
 \den{\wfcxt{\Gamma, x : A}} & =  \compr{\den{\wfcxt{\Gamma}}}{\den{\istype{\Gamma}{A}}} \\
 \den{\wfcxt{\Gamma, \kappa : \clocktype}} & = \compr{\den{\wfcxt{\Gamma}}}{\clk} &
 \den{\wfcxt{\Gamma, \tickA : \kappa}} & =  \tearlier^{\den{\hasnotype\Gamma\kappa{}}}\den{\wfcxt{\Gamma}} 
\end{align*}
\textbf{Types}
\begin{align*}
 \den{\istype{\Gamma}{\latbind\tickA\kappa A}}& =  \tlater^{\den{\hasnotype\Gamma\kappa{}}}\den{\istype{\Gamma, \tickA : \kappa}{A}} &
 \den{\istype\Gamma{\univ{\Delta}}} & = \Usem{\den{\hasnotype\Gamma\Delta{}}} \\
  \den{\istype{\Gamma}{\elems{\Delta}{(t)}}} & = \Elsem{\den{\hasnotype\Gamma\Delta{}}}[\cpair{\id{}}{\den{\hasnotype\Gamma{t}{}}}]
\end{align*}
\textbf{Terms}
\begin{align*}
  \den{\hasnotype{\Gamma, x : A, \Gamma'}x{}} & = \q[\p_{\Gamma'}]  \\
  \den{\hasnotype{\Gamma, \kappa : \clocktype, \Gamma'}\kappa{}} & = \q[\p_{\Gamma'}]  \\
 \den{\hasnotype{\Gamma}{\TickLam{\tickA}{\kappa}A t}{\latbind{\tickA}{\kappa} A}} 
 & = \overline{\den{\hasnotype{\Gamma,\tickA:\kappa}{t}{A}}}  \\
 \den{\hasnotype{\Gamma,\tickA': \kappa,\Gamma'}{\TickApp\tickA\kappa A t{\tickA'}}{A\toksubst{\tickA'}{\tickA}}} 
 & = \overline{\den{\hasnotype{\Gamma}{t}{\latbind{\tickA}{\kappa} A}}}[\p_{\Gamma'}] \\
 \den{\hasnotype{\Gamma}{\TickcApp\tickA \kappa At{\kappa'}}{A\tickcsub\tickA\kappa{\kappa'}}}
 & = \overline{\den{\hasnotype{\Gamma}{t}{\latbind{\tickA}{\kappa} A}}}[d \circ \cpair{\id_{\den\Gamma}}{\den{\hasnotype\Gamma{\kappa'}{}}}] \\
  \den{\hasnotype{\Gamma}{\dfix_A^\kappa\,t}{\later^\kappa A}}  
  & = \dfixebar^{\den{\hasnotype\Gamma\kappa{}}}_{\den\Gamma, \den A}(\den{\hasnotype{\Gamma}{t}{}}) \\
  \den{\hasnotype{\Gamma}{\latbindcodeAnn{\Delta}\tickA\kappa A}{}} 
  & = \code{\tlater^{\den{\hasnotype\Gamma\kappa{}}}_{\den{\hasnotype\Gamma\Delta{}}}}\circ \overline{\den{\hasnotype{\Gamma, \tickA : \kappa}A{}}} \\
  \den{\hasnotype{\Gamma}{\univin{\Delta}{\Delta'}(A)}{}} & = 
  \Usemin{}{\den{\hasnotype\Gamma{\Delta}{}}}{\den{\hasnotype\Gamma{\Delta'}{}}}\circ \den{\hasnotype\Gamma A{}}
\end{align*}
\caption{Interpretation function. The notation $\den{\hasnotype\Gamma{\Delta}{}}$ means
$\{\den{\hasnotype\Gamma{\kappa_1}{}}, \dots \den{\hasnotype\Gamma{\kappa_n}{}}\}$ where 
$\Delta = \{\kappa_1, \dots, \kappa_n\}$. The morphism $d$ is the inverse of $\p_{\tearlier^\q}$
of Proposition~\ref{prop:p:inverse}. }
\label{fig:interp}
\end{center}
\end{figure*}


The most important cases of the partial function defining the interpretation of syntax are given in 
Figure~\ref{fig:interp}. Context, type, and term judgements
are interpreted in the CwF structure of $\grtotal$. The figure excludes the standard cases of $\Pi$- and $\Sigma$-types, 
extensional identity types, as well as universal quantification over clocks, which is modelled as a $\Pi$-type. 
The overlines in the interpretation of terms refer to the bijective correspondence of Theorem~\ref{thm:dradjoint:gr:total}
and the projection $\p_{\Gamma'} : \den{\wfcxt{\Gamma,\tickA': \kappa,\Gamma'}} \to \den{\wfcxt{\Gamma,\tickA': \kappa}}$
in the interpretation of tick application is defined by induction on $\Gamma'$ in the obvious way. 

\subsection{Substitution lemmas}
\label{sec:substitutions}

As is standard in models of dependent type theories \cite{Hofmann1997}, 
welldefinedness of the interpretation function must be proved by induction on the
structure of derivations, simultaneously with soundness and a substitution lemma, 
that we now describe. 

\begin{figure}[tbp]
\begin{center}
\textbf{Formation rules}
\begin{mathpar}
\inferrule*
 { }
 {[] : \Gamma \to \cdot}
\and
\inferrule*
 {\sigma : \Gamma \to \Gamma' \and \hastype{\Gamma}t{A\sigma}}
 {\subex\sigma xt : \Gamma \to \Gamma', x : A}
\and
\inferrule*
 {\sigma : \Gamma \to \Gamma' \and \hastype{\Gamma}{\kappa'}{\clocktype}}
 {\subex\sigma \kappa{\kappa'} : \Gamma \to \Gamma', \kappa : \clocktype}
\and
\inferrule*
 {\sigma : \Gamma \to \Gamma' \and}
 {\subex\sigma{\tickA}{\tickB} : \Gamma, \tickB : \kappa\sigma, \Gamma'' \to \Gamma', \tickA : \kappa}
\and
\inferrule*
 {\sigma : \Gamma \to \Gamma' \and \hastype\Gamma{\kappa'}\clocktype}
 {\subex\sigma{(\tickA : \kappa)}{(\tickc : \kappa')} : \Gamma \to \Gamma', \kappa : \clocktype, \tickA : \kappa} 
\end{mathpar}
\textbf{Interpretation}
\begin{align*}
 \den{[]} & = \, \uniquemap  & 
 \den{\subex\sigma xt} & = \cpair{\den\sigma}{\den t} \\
 \den{\subex\sigma \kappa{\kappa'}} & = \cpair{\den\sigma}{\den{\kappa'}} &
 \den{\subex\sigma{\tickA}{\tickB}} & = \tearlier^{\den{\kappa}}(\den{\sigma}) \circ \p_{\Gamma''} \\
 \den{\subex\sigma{(\tickA : \kappa)}{(\tickc : \kappa')}} & = d \circ \cpair{\den\sigma}{\den{\kappa'}}
\end{align*}
\caption{Rules for wellformedness of syntactic substitutions as well as their interpretation.}
\label{fig:substitutions}
\end{center}
\end{figure}

Figure~\ref{fig:substitutions} lists the rules for wellformedness of syntactic substitutions as well as the definition of the 
partial function interpreting (pre-)substitutions. 
We define the substitution of types and terms along substitutions $\sigma$ in the standard way using the clauses 
of Figure~\ref{fig:clott:typing:annotated} in the cases involving $\tickc$. 

\begin{lem}[Substitution] \label{lem:substitution}
 Let $\sigma : \Gamma \to \Gamma'$ be a wellformed substitution then $\den\sigma$ is a 
 welldefined morphism
 from $\den\Gamma$ to $\den{\Gamma'}$. If $\istype{\Gamma'} A$ then also
 $\istype{\Gamma}{A\sigma}$ and $\den{\istype\Gamma {A\sigma}} = \den{\istype{\Gamma'}A}[\den{\sigma}]$. If
 $\hastype{\Gamma'} tA$ then $\hastype{\Gamma}{t\sigma}{A\sigma}$ and $\den{\hasnotype{\Gamma}{t\sigma}{}} 
 = \den{\hasnotype{\Gamma'}t{}}[\den\sigma]$. 
\end{lem}

Before proving Lemma~\ref{lem:substitution} we need to establish a few facts about projections. In particular,
we need Lemma~\ref{lem:subst:comp:proj} below for the case of variable introduction. Lemma~\ref{lem:subst:comp:proj}
is trivial in the setting of standard type theory, but requires a little more work in a setting with
ticks and tick weakening as in \clott. 

\begin{lem} \label{lem:syntactic:proj}
 If $\Gamma, \Gamma'$ is a wellformed context, the obvious context projection 
 $\wksubst{\Gamma}{\Gamma'} : \Gamma, \Gamma' \to \Gamma$ is also wellformed, and  
 $\den{\wksubst{\Gamma}{\Gamma'}} = \p_{\Gamma'}$.
\end{lem}

\begin{proof}
 By induction on the height of $\Gamma$. If $\Gamma$ is of length $0$ then the statement is trivial.
 If $\Gamma = \Gamma_0, x : A$, then by induction the projection 
 $\Gamma_0, x : A, \Gamma' \to\Gamma_0$ is wellformed and interpreted as 
 $\p_{x: A, \Gamma'} = \p \circ \p_{\Gamma'}$. Now, the projection 
 $\Gamma_0, x : A, \Gamma' \to\Gamma_0, x : A$ is interpreted as 
 $\cpair{\p \circ \p_{\Gamma'}}{\q[\p_{\Gamma'}]} = \cpair{\p}{\q} \circ \p_{\Gamma'} = \p_{\Gamma'}$.
 The case of extension of $\Gamma$ with $\kappa : \clocktype$ is similar. 
 In the case of $\Gamma = \Gamma_0, \tickA : \kappa$, by induction, the syntactic projection 
 $\Gamma_0 \to \Gamma_0$ is well defined and interpreted as the identity. 
 Now, the projection $\Gamma_0, \tickA : \kappa, \Gamma' \to \Gamma_0, \tickA : \kappa$
 is by definition $\tearlier^\den{\kappa}(\id) \circ \p_{\Gamma'} = \p_{\Gamma'}$.
\end{proof}

If $\sigma : \Gamma \to \Gamma'$ and $\Gamma, \Gamma_0$ are well-formed we define the weakening
of $\sigma$ to be the substitution $\sigma\circ \wksubst{\Gamma}{\Gamma_0} : \Gamma, \Gamma_0 \to \Gamma'$
as follows
\begin{align*}
 []\circ \wksubst{\Gamma}{\Gamma_0} & = [] \\
 (\subex\sigma xt) \circ \wksubst{\Gamma}{\Gamma_0} 
 & = \subex{(\sigma \circ \wksubst{\Gamma}{\Gamma_0})}x{t\,\wksubst{\Gamma}{\Gamma_0}} \\
 (\subex\sigma \kappa{\kappa'}) \circ \wksubst{\Gamma}{\Gamma_0} 
 & = \subex{(\sigma \circ \wksubst{\Gamma}{\Gamma_0})}\kappa{\kappa'\wksubst{\Gamma}{\Gamma_0}} \\
 (\subex\sigma \tickA \tickB) \circ \wksubst{\Gamma}{\Gamma_0} & = \subex\sigma \tickA \tickB \\
 \subex\sigma{(\tickA : \kappa)}{(\tickc : \kappa')} \circ \wksubst{\Gamma}{\Gamma_0} 
 & = \subex{(\sigma\circ \wksubst{\Gamma}{\Gamma_0})}{(\tickA : \kappa)}{(\tickc : \kappa'\wksubst{\Gamma}{\Gamma_0})}
\end{align*}

\begin{lem} \label{lem:proj:comp:subst}
 If $\sigma : \Gamma \to \Gamma'$ and $\Gamma, \Gamma_0$ are wellformed, so is 
 $\sigma\circ \wksubst{\Gamma}{\Gamma_0}$ and 
 $\den{\sigma\circ \wksubst{\Gamma}{\Gamma_0}} = \den\sigma \circ \p_{\Gamma_0}$
\end{lem}

\begin{proof}
 An easy induction on $\sigma$ using Lemma~\ref{lem:substitution}. 
\end{proof}

The next lemma refers to the notion of prefix of a substitution, which is the reflexive-transitive closure of the following rules
\begin{align*}
  \sigma & \leq \subex\sigma xt & 
  \sigma & \leq \subex\sigma \kappa{\kappa'} \\
  \sigma\circ\wksubst{\Gamma}{\tickB : \kappa\sigma, \Gamma''} & \leq \subex\sigma{\tickA}{\tickB}  & 
  \subex\sigma \kappa{\kappa'} & \leq \subex\sigma{(\tickA : \kappa)}{(\tickc : \kappa')}
\end{align*}
In the case of $\subex\sigma{\tickA}{\tickB}$, the weakening $\wksubst{\Gamma}{\tickB : \kappa\sigma, \Gamma''}$
refers to the metavariables as in the formation rule for $\subex\sigma{\tickA}{\tickB}$. 

\begin{lem} \label{lem:subst:comp:proj}
 If $\sigma : \Gamma \to \Gamma_0', \Gamma_1'$ is a syntactic substitution, then there is a 
 prefix $\tau$ of $\sigma$ such that $\tau : \Gamma \to \Gamma_0'$ and 
 $\p_{\Gamma_1'} \circ \den\sigma = \den\tau$. 
\end{lem}

\begin{proof}
 By induction on the length of $\Gamma_1'$ simultaneously with the proof of 
 Lemma~\ref{lem:substitution}. The case of $\Gamma_1'$ empty is trivial. If 
 $\Gamma_1' = x : A, \Gamma_1''$ then by induction, there is a prefix $\tau$
 of $\sigma$ such that $\den \tau = \p_{\Gamma_1''} \circ\den \sigma$. By inspection,
 $\tau$ must be of the form $\subex{\tau'}xt$, and 
 \[
  \p_{\Gamma_1'} \circ\den \sigma = \p \circ \p_{\Gamma_1''} \circ\den \sigma 
  = \p\circ \cpair{\den {\tau'}}{\den t} = \den{\tau'}
 \]
 In the case that $\Gamma_1' = \tickA : \kappa, \Gamma_1''$, again by induction, 
 there is a prefix $\tau$
 of $\sigma$ such that $\den \tau = \p_{\Gamma_1''} \circ\den \sigma$. Now, there are
 two cases for $\tau$. One is that $\tau = \subex{\tau'}{\tickA}{\tickB}$. In this case
 $\Gamma$ must be of the form $\Gamma_0, \tickB : \kappa[\sigma], \Gamma_1$ and 
\begin{align*}
  \p_{\Gamma_1'} \circ\den \sigma & = \p_{\tearlier^{\den\kappa}} \circ \p_{\Gamma_1''} \circ\den \sigma \\
  & = \p_{\tearlier^{\den\kappa}} \circ \tearlier^{\den\kappa}(\den{\tau'}) \circ \p_{\Gamma_1} \\
  & = \den{\tau'} \circ\p_{\tearlier^{\den\kappa}} \circ \p_{\Gamma_1} \\
  & = \den{\tau'\circ \wksubst{\Gamma_0}{\tickB : \kappa[\sigma], \Gamma_1}}
\end{align*}
The other case for $\tau$ is $\tau = \subex{\tau'}{(\tickA : \kappa)}{(\tickc : \kappa')}$. In this case 
$\subex{\tau'}\kappa{\kappa'}$ is a prefix of $\sigma$ and 
\begin{align*}
 \p_{\Gamma_1'} \circ\den \sigma & = \p_{\tearlier^{\den\kappa}} \circ \p_{\Gamma_1''} \circ\den \sigma  \\
  & = \p_{\tearlier^{\den\kappa}} \circ d \circ \cpair{\den{\tau'}}{\den{\kappa'}} \\
  & = \cpair{\den{\tau'}}{\den{\kappa'}} \\
  & = \den{\subex{\tau'}\kappa{\kappa'}}
\end{align*}
since $d$ is the inverse of $\p_{\tearlier^{\den\kappa}}$.
\end{proof}

\begin{proofof}{Lemma~\ref{lem:substitution}}
 As in~\cite{Hofmann1997} this is proved by induction on the sizes of the terms, contexts, and 
 types involved, simultaneously with Theorem~\ref{thm:soundness} and the three lemmas above. 
 Most cases are standard, including the cases for universal quantification over clocks which is modelled
 here as a $\Pi$-type. We just explain the non-standard cases. 

We start with the case of variable introduction. Suppose $\Gamma' = \Gamma'_0, x : A, \Gamma_1'$, and 
$t = x$. By Lemma~\ref{lem:subst:comp:proj}, there must be a prefix 
$\subex\tau xu : \Gamma \to \Gamma'_0, x : A$ of $\sigma$ such that 
$\p_{\Gamma_1'}\circ \den\sigma = \den{\subex\tau xu}$. Now,
\begin{align*}
 \den{\hastype{\Gamma'_0, x : A, \Gamma_1'}xA}[\den\sigma] & = 
 \q[\p_{\Gamma_1'}\circ \den\sigma] = \q[\den{\subex\tau xu}] = \den{\hasnotype{\Gamma}{u}{A\sigma}} \, .
\end{align*}
The case of clock introduction is similar.  

The case of $\latbind\tickA\kappa A$ is proved as follows:
\begin{align*}
 \den{\hasnotype\Gamma{\latbind\tickA\kappa A}{}}[\den\sigma]
 & =  (\tlater^{\den\kappa}\den{\hasnotype{\Gamma, \tickA : \kappa}{A}{}})[\den\sigma] \\
 & =  \tlater^{\den\kappa[\den\sigma]}(\den{\hasnotype{\Gamma, \tickA : \kappa}{A}{}}[\tearlier^{\den\kappa}\den\sigma]) \\
 & =  \tlater^{\den{\kappa\sigma}}(\den{\hasnotype{\Gamma, \tickA : \kappa}{A}{}}[\tearlier^{\den\kappa}\den\sigma]) \\
 & =  \tlater^{\den{\kappa\sigma}}(\den{\hasnotype{\Gamma, \tickA : \kappa}{A}{}}[\den{\subex\sigma\tickA\tickA}]) \\
 & =  \tlater^{\den{\kappa\sigma}}\den{\hasnotype{\Gamma', \tickA : \kappa\sigma}{A\subex\sigma\tickA\tickA}{}} \\
 & =  \den{\hasnotype{\Gamma'}{\latbind\tickA{\kappa\sigma} (A\subex\sigma\tickA\tickA}{})} \\
 & =  \den{\hasnotype{\Gamma'}{(\latbind\tickA\kappa A)\sigma}{})} 
\end{align*}
using Theorem~\ref{thm:dradjoint:gr:total}.\ref{item:later:nat} in the second equality. The proof of the case of 
$\TickLam{\tickA}{\kappa}A t$ is very similar, using Theorem~\ref{thm:dradjoint:gr:total}.\ref{item:overline:nat}.
 
 Case of $\hastype{\Gamma_0',\tickB: \kappa,\Gamma'_1}{\TickApp\tickA\kappa A t{\tickB}}{A\toksubst{\tickB}{\tickA}}$: In this case the 
 typing assumption is $\hastype{\Gamma}{t}{\latbind{\tickA}{\kappa}}A$ and 
\begin{align*}
 \den{\TickApp\tickA\kappa A t{\tickB}}[\den\sigma] & = \overline{\den t}[\p_{\Gamma_1'}\circ \den \sigma] = \overline{\den t}\den\tau
\end{align*}
for some prefix $\tau : \Gamma \to \Gamma_0', \tickA: \kappa$ of $\sigma$ by Lemma~\ref{lem:subst:comp:proj}. Now, there
are two possible cases for $\tau$. The first and simplest case is $\tau = \subex{\tau'}\tickB{\tickB'}$, where 
$\Gamma = \Gamma_0, \tickB' : \kappa\tau, \Gamma_1$ and $\tau' : \Gamma_0 \to \Gamma_0'$. In this case 
\begin{align*}
 \overline{\den t}\den\tau & = \overline{\den t}[\tearlier^{\den\kappa}(\den{\tau'}) \circ \p_{\Gamma_1}] \\
 & =  \overline{\den t[\den{\tau'}]}[\p_{\Gamma_1}] \\ 
 & =  \overline{\den {t\tau'}}[\p_{\Gamma_1}] \\ 
 & =  \den{\TickApp\tickA{\kappa\tau'}{A\tau'} {t\tau'}{\tickB'}} \\
 & =  \den{(\TickApp\tickA{\kappa}{A}t\tickB)\sigma} 
\end{align*}
using Theorem~\ref{thm:dradjoint:gr:total}.\ref{item:overline:nat} for the second equality and the induction hypothesis for the third. 

The second case for $\tau$ is $\tau = \subex{\tau'}{(\tickB : \kappa)}{(\tickc : \kappa')}$. This requires 
$\Gamma_0' = \Gamma'', \kappa : \clocktype$ for some $\Gamma''$. In this case 
\begin{align*}
 \overline{\den t}\den\tau & = \overline{\den t}[d \circ \cpair{\den{\tau'}}{\den{\kappa'}}] \\
 & = \overline{\den t}[d \circ \cpair{\den{\tau'}\circ \p}{\q} \circ \cpair{\id}{\den{\kappa'}}] \\
 & = \overline{\den t}[\tearlier^\q\cpair{\den{\tau'}\circ \p}{\q} \circ d \circ \cpair{\id}{\den{\kappa'}}] \\
 & = \overline{\den t[\cpair{\den{\tau'}\circ \p}{\q}]}[d \circ \cpair{\id}{\den{\kappa'}}] 
\end{align*}
 using naturality of $d$ (which follows from $\p_{\tearlier^\q}$ being natural) and 
 Theorem~\ref{thm:dradjoint:gr:total}.\ref{item:overline:nat}. Now, by Lemma~\ref{lem:proj:comp:subst}
 $\den{\tau'}\circ \p = \den{\tau' \circ \wksubst{\Gamma''}{\kappa : \clocktype}}$, and so (not writing the weakening
 for simplicity) $\cpair{\den{\tau'}\circ \p}{\q} = \den{\subex{\tau'}{\kappa}{\kappa}}$. Thus
\begin{align*}
 \overline{\den t}\den\tau
 & = \overline{\den{t\subex{\tau'}{\kappa}{\kappa}}}[d \circ \cpair{\id}{\den{\kappa'}}] \\
 & = \den{\TickcApp\tickA \kappa {(A\subex{\subex{\tau'}{\kappa}{\kappa}}{\tickA}{\tickA})}{(t\subex{\tau'}{\kappa}{\kappa})}{\kappa'}} \\
 & = \den{(\TickApp\tickA \kappa {A}{t}\tickB)\subex{\tau'}{(\tickB : \kappa)}{(\tickc : \kappa')}} \\
 & = \den{(\TickApp\tickA \kappa {A}{t}\tickB)\sigma}
\end{align*}
  
 Case of $\TickcApp\tickA \kappa At{\kappa'}$: In this case, the assumptions on the typing rule state that 
 $\hastype[]{\Gamma, \kappa : \clocktype}{t}{\latbind{\tickA}{\kappa} A}$ and  $\hastype{\Gamma}{\kappa'}\clocktype$.
 The case is proved as follows
\begin{align*}
 \den{\TickcApp\tickA \kappa At{\kappa'}}[\den\sigma] & = \overline{\den t}[d \circ \cpair{\id}{\den{\kappa'}} \circ \den\sigma] \\
 & = \overline{\den t}[d \circ \cpair{\den\sigma\circ\p}{\q} \circ \cpair{\id}{\den{\kappa'}[\den\sigma]}]  \\
 & = \overline{\den t}[\tearlier^{\den\kappa}\cpair{\den\sigma\circ\p}{\q} \circ d \circ \cpair{\id}{\den{\kappa'\sigma}}]  \\
 & = \overline{\den t[\den{\subex{\sigma}\kappa\kappa}]}[d \circ \cpair{\id}{\den{\kappa'\sigma}}]  \\
 & = \overline{\den {t(\subex{\sigma}\kappa\kappa)}}[d \circ \cpair{\id}{\den{\kappa'\sigma}}]  \\
 & = \den{\TickcApp\tickA{\kappa} {A(\subex{\subex\sigma\kappa\kappa}\tickA\tickA)}{(t(\subex\sigma\kappa\kappa))}{\kappa'\sigma}} \\
 & =  \den{(\TickcApp\tickA \kappa At{\kappa'})\sigma}
\end{align*}
The case of $\dfix_A^\kappa\,t$ follows from Lemma~\ref{lem:fix:points}. The case of universes 
follows from Lemma~\ref{lem:univ:reindexing} and  the case of 
$\latbindcodeAnn{\Delta}\tickA\kappa A$ follows from 
Theorem~\ref{thm:tlater:code}.\ref{item:tlater:code:2}:
\begin{align*}
 \den{\latbindcodeAnn{\Delta}\tickA\kappa A} [\den \sigma]
 & = (\code{\tlater^{\den{\kappa}}_{\den{\Delta}}}\circ \overline{\den{A}}) [\den\sigma] \\
 & = \code{\tlater^{\den{\kappa}}_{\den{\Delta}}}[\den\sigma]\circ \overline{\den{A}}[\den\sigma]  \\
 & = \code{\tlater^{\den{\kappa}[\den\sigma]}_{\den{\Delta}[\den\sigma]}}\circ \overline{\den{A}[\den{\subex\sigma\tickA\tickA}]}  \\
 & = \code{\tlater^{\den{\kappa\sigma}}_{\den{\Delta\sigma}}}\circ \overline{\den{A\subex\sigma\tickA\tickA}} \\
 & = \den{\latbindcodeAnn{(\Delta\sigma)}\tickA{\kappa\sigma} (A\sigma)} \\
 & = \den{(\latbindcodeAnn{\Delta}\tickA\kappa A)\sigma}
\end{align*}
The cases of other codes on the universe follow from similar 
theorems found in~\cite{GDTTmodel}. 
\end{proofof}

\subsection{Soundness}

The main theorem of the paper is the following.
\begin{thm}\label{thm:soundness}
 If a judgement has a derivation, then the interpretation of it is well defined. If $\hastype{\Gamma}tA$ has a derivation
 then $\den{\hasnotype\Gamma tA}$ is an element of the family $\den{\istype\Gamma A}$. 
 Moreover, the interpretation is sound with respect to the
 equalities of Figure~\ref{fig:clott:typing} and Figure~\ref{fig:universes} and models the 
 axioms (\ref{eq:cirr}), (\ref{eq:tirr}) and (\ref{eq:pfix}). 
\end{thm}

For the case of the clock irrelevance axiom we need the following lemma.

\begin{lem} \label{lem:inv:clock:intro}
 The interpretation of any type, if well-defined, is invariant under clock introduction in the sense of Definition~\ref{def:inv:clock:intro}.
\end{lem}

\begin{proof}
 By an easy induction on the structure of types using Lemma~\ref{lem:closure:iuci} in most cases and arguments as in the proof
 of Lemma~\ref{thm:tlater:code} in the case of $\latbind\tickA \kappa A$.
\end{proof}

\begin{proofof}{Theorem~\ref{thm:soundness}}
 As mentioned above, the proof is by induction on typing judgements, and we just show the cases listed in 
 Figure~\ref{fig:clott:typing:annotated}. The cases of tick abstraction and application follow straightforwardly from the bijection of Theorem~\ref{thm:dradjoint:gr:total}.\ref{item:overline:nat}.
 In the case of application, Lemma~\ref{lem:substitution} and Lemma~\ref{lem:syntactic:proj} ensure that the element
 $\overline{\den{\hasnotype{\Gamma}{t}{\latbind{\tickA}{\kappa} A}}}[\p_{\Gamma'}]$, which a priori is an element of 
 $\den{\istype{\Gamma, \tickB: \kappa}{A\subst{\tickA}{\tickB}}}[\p_{\Gamma'}]$, is also an element of 
 $\den{\istype{\Gamma, \tickB: \kappa, \Gamma'}{A\subst{\tickA}{\tickB}}}$ as required. In the case of 
 application to $\tickc$, the expression
 $\overline{\den{\hasnotype{\Gamma}{t}{\latbind{\tickA}{\kappa} A}}}[d \circ \cpair{\id_{\den\Gamma}}{\den{\kappa'}}]$
 is by induction an element of $\den{\istype{\Gamma, \kappa : \clocktype, \tickA: \kappa}A}[d \circ \cpair{\id_{\den\Gamma}}{\den{\kappa'}}]$,
 which, since $d \circ \cpair{\id_{\den\Gamma}}{\den{\kappa'}} = \den{\subex{\id_\Gamma}{(\tickA, \kappa)}{(\tickc, \kappa')}}$, 
 by Lemma~\ref{lem:substitution} equals $\den{\istype{\Gamma}{A\tickcsub\tickA\kappa{\kappa'}}}$
  
 Most equalities are straightforward. For example, the $\eta$-rules for clock and tick abstraction follow from the bijective
 correspondence on elements in dependent right adjoints. The $\beta$-rules for these are similar, but involve the substitution
 lemma in the case of clocks and a simple check that the interpretation is invariant under renaming of ticks in the case of ticks.
 The case of the $\beta$-rule for application to $\tickc$, i.e., equation (\ref{eq:annotated:beta}), follows directly from  
 Lemma~\ref{lem:substitution}. 
 The judgemental equality for fixed point unfolding at $\tickc$ follows from the fixed point unfolding axiom (\ref{eq:pfix}), 
 since the model is extensional. To prove the latter it suffices to show that the interpretations of
 $\tabs\tickA\kappa{t(\dfix^\kappa t)}$ and $\dfix^\kappa t$ are equal, which follows from 
 Lemma~\ref{lem:fix:points}. For the interpretation of the tick irrelevance axiom (\ref{eq:tirr}) it suffices to show that 
 the interpretations of $\hastype{\Gamma, \tickA : \kappa, \tickA' : \kappa}{\tapp t}A$ and
 $\hastype{\Gamma, \tickA : \kappa, \tickA' : \kappa}{\tapp[\tickA'] t }A$ are equal. By definition 
 $\den{\tapp t} = \overline{\den t}[\p_{\tearlier^{\den{\kappa}}}]$ and by the substitution lemma 
 $\den{\tapp[\tickA'] t} = \overline{\den t[\p_{\tearlier^{\den{\kappa}}}]}$ which equals 
 $\overline{\den t}[\tearlier^{\den{\kappa}} (\p_{\tearlier^{\den{\kappa}}})]$, and so the equality
 follows from Lemma~\ref{lem:tirr}. The soundness of the clock irrelevance axiom (\ref{eq:cirr})
 follows from Lemma~\ref{lem:inv:clock:intro} as in~\cite{GDTTmodel}.
 
 The equalities in Figure~\ref{fig:universes} were proved sound in \cite{GDTTmodel}, 
 except the two involving $\latbindcode{\tickA}{\kappa}(-)$. For the first of these,
 note that by definition
\begin{align*}
 \den{\hasnotype{\Gamma, \tickA : \kappa}{\elems{\Delta}(A)}{}} & = 
 \Elsem{\den{\hasnotype{\Gamma, \tickA : \kappa}\Delta{}}}[\cpair{\id{}}{\den{\hasnotype{\Gamma, \tickA : \kappa}{A}{}}}] \\
 & \Elsem{\den{\hasnotype{\Gamma}\Delta{}}[\p_{\tearlier^{\den\kappa}}]}[\cpair{\id{}}{\den{\hasnotype{\Gamma, \tickA : \kappa}{A}{}}}]
\end{align*}
 and so by Theorem~\ref{thm:tlater:code}.\ref{item:tlater:code:1} we get
 \[
  \Elsem{\den{\hasnotype{\Gamma}\Delta{}}}[\cpair{\id}{\code{\tlater^{\den{\hasnotype{\Gamma}\kappa{}}}_{\den{\hasnotype{\Gamma}\Delta{}}}} \circ \overline{\den{\hasnotype{\Gamma, \tickA : \kappa}{A}{}}}}] 
 = \tlater^{\den{\hasnotype{\Gamma}\kappa{}}} 
 \den{\hasnotype{\Gamma, \tickA : \kappa}{\elems\Delta(A)}{}}
 \]
 From this we deduce
 \begin{align*}
  \den{\elems{\Delta}{(\latbindcodeAnn\Delta \tickA \kappa A)}} 
  & =  \den{\latbind\tickA\kappa \elems{\Delta}(A)} 
\end{align*}
by unfolding definitions.
 
 The second of these can be proved using  
 Theorem~\ref{thm:tlater:code}.\ref{item:tlater:code:3} as we now show. 
\begin{align*}
   \den{\univin{{\Delta}}{\Delta'}{(\latbindcodeAnn{\Delta} \tickA \kappa A)}} 
   & = \Usemin{}{\den\Delta}{\den{\Delta'}} \circ \code{\tlater^{\den{\kappa}}_{\den{\Delta}}}\circ \overline{\den A} \\
   & = \code{\tlater^{\den\kappa}_{\den{\Delta'}}}\circ 
   \overline{(\Usemin{}{(\den{\Delta}[\p_{\tlater^{\den\kappa}}])}{(\den{\Delta'}[\p_{\tlater^{\den\kappa}}])} \circ \den A)} \\
   & = \code{\tlater^{\den\kappa}_{\den{\Delta'}}}\circ 
   \overline{(\Usemin{}{\den{\Delta}}{\den{\Delta'}} \circ \den A)} \\
   & =\den{\latbindcodeAnn{\Delta'} \tickA \kappa{\univin{\Delta}{\Delta'}{(A)}}} 
\end{align*}
where the third equality uses $\den{\hasnotype\Gamma\Delta{}}[\p_{\tlater^{\den\kappa}}]
= \den{\hasnotype{\Gamma, \tickA : \kappa}\Delta{}}$ which follows from 
Lemma~\ref{lem:substitution}. 
\end{proofof}

\section{Conclusion and future work}
\label{sec:conclusion}

As mentioned in the introduction, the model presented here improves on the previous model of
\clott~\cite{conferenceversion} not only by fixing a mistake, but also by
greatly simplifying it. The simplification is made possible by using a single-context presentation 
of the syntax and an internal indexing of the dependent right adjoints over the clock object. 

The typing rule for application to $\tickc$ as presented in~\cite{bahr2017clocks} and in Section~\ref{sec:clott}
differs from the one of the elaborated syntax interpreted in the model, as presented in Section~\ref{sec:interp:syntax} by
replacing a substitution by a delayed one. This means that for terms $t,s$ such that 
$t \subst\kappa{\kappa'} = s \subst\kappa{\kappa'}$ the application of $t\subst\kappa{\kappa'}$ 
and $s\subst\kappa{\kappa'}$ to $\tickc$ in the original
syntax are literally the same, whereas they are different in the elaborate syntax, and could potentially
be different in the model. This means that the interpretation of a term $\tappc{t\subst\kappa{\kappa'}}$ involves a choice of $s$ or $t$ 
as above. We view this choice as part of the elaboration of terms, similarly to the choice of $\Pi$ type for 
application terms $t\, u$. 

We suspect that in most cases the choice mentioned above will not affect the interpretation of
a term, but
our attempts at proving that have failed. In particular, in the situation above, one can construct
a term $u$ such that $u\subst{\kappa, \kappa'}{\kappa_0,\kappa_1} = t$ and 
$u\subst{\kappa, \kappa'}{\kappa_1,\kappa_0} = s$ and use this to prove that the interpretations are equal. However,
it is unclear if $u$ is welltyped, and so the interpretation could fail to be welldefined.

The dependent right adjoints considered in this paper are all endo-adjunctions, i.e., the domain
and codomain of the left adjoint are the same. Multimodal Dependent Type Theory~\cite{gratzer2020multimodal} 
is a recent extension of the idea of Fitch-style modal types allowing also dependent adjunctions between 
different categories, and even 2-dimensional diagrams of these, referred to as \emph{mode theories}. 
This is a very general framework capturing many dependent type theories with modal operators, 
but since the parametrisation by mode theory is given externally, we suspect that the single context
presentation of \clott\ falls outside this framework. It would be interesting to see if there is a 
generalisation of Multimodal Dependent Type Theory that captures also single context \clott.

Our motivation for constructing this model is to study extensions of \clott. In particular, we would like to extend \clott\ with
path types as in~\cite{GCTT}. This requires an adaptation of the model to the cubical setting, 
using $\catT$-indexed
families of cubical sets~\cite{CTT} rather than just sets. 
The resulting variant of \clott\ should be closer to the intensional type theory presented 
in~\cite{bahr2017clocks} than the extensional type theory modelled here. For example, fixed 
points should unfold only up to path equality. We also hope that formulating the clock irrelevance
axiom using path equality the universes can be modelled differently, simplifying the perhaps 
most technical part of the present model construction.  

\subsection*{Acknowledgements.} We thank Patrick Bahr for useful discussions. 

\bibliographystyle{alpha}
\bibliography{paper.bib}

\newcommand{\etalchar}[1]{$^{#1}$}
\begin{thebibliography}{CMM{\etalchar{+}}20}

\bibitem[AM01]{Appel:M01}
Andrew~W. Appel and David McAllester.
\newblock An indexed model of recursive types for foundational proof-carrying
  code.
\newblock {\em ACM Trans. Program. Lang. Syst}, 23(5):657--683, 2001.

\bibitem[AM13]{atkey13icfp}
Robert Atkey and Conor McBride.
\newblock Productive coprogramming with guarded recursion.
\newblock {\em ACM SIGPLAN Notices}, 48(9):197--208, 2013.

\bibitem[AP13]{Abel:Wellfounded}
Andreas Abel and Brigitte Pientka.
\newblock Wellfounded recursion with copatterns: A unified approach to
  termination and productivity.
\newblock In {\em Proceedings ICFP 2013}, pages 185--196. ACM, 2013.

\bibitem[AVW17]{Abel:NBE:sized:types}
Andreas Abel, Andrea Vezzosi, and Theo Winterhalter.
\newblock Normalization by evaluation for sized dependent types.
\newblock {\em Proceedings of the ACM on Programming Languages}, 1(ICFP):1--30,
  2017.

\bibitem[Awo18]{awodey2018natural}
Steve Awodey.
\newblock Natural models of homotopy type theory.
\newblock {\em Mathematical Structures in Computer Science}, 28(2):241--286,
  2018.

\bibitem[BBC{\etalchar{+}}19]{GCTT}
Lars Birkedal, Ale\v{s} Bizjak, Ranald Clouston, Hans~Bugge Grathwohl, Bas
  Spitters, and Andrea Vezzosi.
\newblock Guarded cubical type theory: Path equality for guarded recursion.
\newblock {\em Journal of Automated Reasoning}, 63(2):211--253, 2019.

\bibitem[BBM14]{Bizjak-et-al:countable-nondet-internal}
Ale\v{s} Bizjak, Lars Birkedal, and Marino Miculan.
\newblock A model of countable nondeterminism in guarded type theory.
\newblock In {\em Rewriting and Typed Lambda Calculi}, pages 108--123.
  Springer, 2014.

\bibitem[BCM15]{bernardy2015presheaf}
Jean-Philippe Bernardy, Thierry Coquand, and Guilhem Moulin.
\newblock A presheaf model of parametric type theory.
\newblock {\em Electronic Notes in Theoretical Computer Science}, 319:67--82,
  2015.

\bibitem[BGBC17]{birkedal2017guarded}
Lars Birkedal, Hans~Bugge Grathwohl, Ale{\v{s}} Bizjak, and Ranald Clouston.
\newblock The guarded lambda-calculus: Programming and reasoning with guarded
  recursion for coinductive types.
\newblock {\em Logical Methods in Computer Science}, 12, 2017.

\bibitem[BGC{\etalchar{+}}16]{GDTT}
Ale{\v{s}} Bizjak, Hans~Bugge Grathwohl, Ranald Clouston, Rasmus~E
  M{\o}gelberg, and Lars Birkedal.
\newblock Guarded dependent type theory with coinductive types.
\newblock In {\em International Conference on Foundations of Software Science
  and Computation Structures}, pages 20--35. Springer, 2016.

\bibitem[BGM17]{bahr2017clocks}
Patrick Bahr., Hans~Bugge Grathwohl, and Rasmus~Ejlers M{\o}gelberg.
\newblock The clocks are ticking: No more delays!
\newblock In {\em 2017 32nd Annual ACM/IEEE Symposium on Logic in Computer
  Science (LICS)}, pages 1--12. IEEE, 2017.

\bibitem[BM20]{GDTTmodel}
Ale\v{s} Bizjak and Rasmus~Ejlers M{\o}gelberg.
\newblock Denotational semantics for guarded dependent type theory.
\newblock {\em Mathematical Structures in Computer Science}, 30(4):342--378,
  2020.

\bibitem[BMSS12]{Birkedal-et-al:topos-of-trees}
Lars Birkedal, Rasmus~Ejlers M{\o}gelberg, Jan Schwinghammer, and Kristian
  St{\o}vring.
\newblock First steps in synthetic guarded domain theory: step-indexing in the
  topos of trees.
\newblock {\em Logical Methods in Computer Science}, 8(4), 2012.

\bibitem[CCHM18]{CTT}
Cyril Cohen, Thierry Coquand, Simon Huber, and Anders M{\"o}rtberg.
\newblock Cubical type theory: A constructive interpretation of the univalence
  axiom.
\newblock In {\em 21st International Conference on Types for Proofs and
  Programs (TYPES 2015)}. Schloss Dagstuhl-Leibniz-Zentrum fuer Informatik,
  2018.

\bibitem[Clo18]{clouston2018fitch}
Ranald Clouston.
\newblock Fitch-style modal lambda calculi.
\newblock In {\em International Conference on Foundations of Software Science
  and Computation Structures}, pages 258--275. Springer, 2018.

\bibitem[CMM{\etalchar{+}}20]{drat}
Ranald Clouston, Bassel Mannaa, Rasmus~Ejlers M{\o}gelberg, Andrew~M. Pitts,
  and Bas Spitters.
\newblock Modal dependent type theory and dependent right adjoints.
\newblock {\em Mathematical Structures in Computer Science}, 30(2):118--138,
  2020.

\bibitem[Coq93]{coquand1993infinite}
Thierry Coquand.
\newblock Infinite objects in type theory.
\newblock In {\em International Workshop on Types for Proofs and Programs},
  pages 62--78. Springer, 1993.

\bibitem[Dyb95]{dybjer1996}
Peter Dybjer.
\newblock Internal type theory.
\newblock In {\em International Workshop on Types for Proofs and Programs},
  pages 120--134. Springer, 1995.

\bibitem[Fit52]{Fitch:Symbolic}
Frederic~Benton Fitch.
\newblock {\em Symbolic logic, an introduction}.
\newblock Ronald Press Co., New York, NY, USA, 1952.

\bibitem[GKNB20]{gratzer2020multimodal}
Daniel Gratzer, GA~Kavvos, Andreas Nuyts, and Lars Birkedal.
\newblock Multimodal dependent type theory.
\newblock In {\em Proceedings of the 35th Annual ACM/IEEE Symposium on Logic in
  Computer Science}, pages 492--506, 2020.

\bibitem[Hof97]{Hofmann1997}
Martin Hofmann.
\newblock Syntax and semantics of dependent types.
\newblock In {\em Extensional Constructs in Intensional Type Theory}, pages
  13--54. Springer, 1997.

\bibitem[HPS96]{HughesPS96}
J.~Hughes, L.~Pareto, and A.~Sabry.
\newblock Proving the correctness of reactive systems using sized types.
\newblock In {\em Conference Record of POPL'96: The 23rd {ACM} {SIGPLAN-SIGACT}
  Symposium on Principles of Programming Languages, Papers Presented at the
  Symposium, St. Petersburg Beach, Florida, USA, January 21-24, 1996}, pages
  410--423, 1996.

\bibitem[HS99]{Hofmann-Streicher:lifting}
Martin Hofmann and Thomas Streicher.
\newblock Lifting {G}rothendieck universes.
\newblock Unpublished, 1999.

\bibitem[ML84]{MartinLof:84}
Per Martin-L\"of.
\newblock {\em Intuitionistic Type Theory}.
\newblock Bibliopolis, Napoli, 1984.

\bibitem[MM18]{conferenceversion}
Bassel Mannaa and Rasmus~Ejlers M{\o}gelberg.
\newblock The clocks they are adjunctions: Denotational semantics for clocked
  type theory.
\newblock In {\em 3rd International Conference on Formal Structures for
  Computation and Deduction, {FSCD} 2018, July 9-12, 2018, Oxford, {UK}}, pages
  23:1--23:17, 2018.

\bibitem[MP08]{mcbride2008applicative}
Conor McBride and Ross Paterson.
\newblock Applicative programming with effects.
\newblock {\em Journal of functional programming}, 18(1):1--13, 2008.

\bibitem[Nak00]{Nakano:Modality}
Hiroshi Nakano.
\newblock A modality for recursion.
\newblock In {\em Proceedings Fifteenth Annual IEEE Symposium on Logic in
  Computer Science}, pages 255--266. IEEE, 2000.

\bibitem[PMD15]{FreshMLTT}
Andrew~M. Pitts, Justus Matthiesen, and Jasper Derikx.
\newblock A dependent type theory with abstractable names.
\newblock {\em Electronic Notes in Theoretical Computer Science}, 312:19--50,
  2015.

\bibitem[Sac13]{Sacchini13}
Jorge~Luis Sacchini.
\newblock Type-based productivity of stream definitions in the calculus of
  constructions.
\newblock In {\em 28th Annual {ACM/IEEE} Symposium on Logic in Computer
  Science, {LICS} 2013, New Orleans, LA, USA, June 25-28, 2013}, pages
  233--242, 2013.

\bibitem[SH18]{sterling2018guarded}
Jonathan Sterling and Robert Harper.
\newblock Guarded computational type theory.
\newblock In {\em Proceedings of the 33rd Annual ACM/IEEE Symposium on Logic in
  Computer Science}, pages 879--888, 2018.

\bibitem[{Uni}13]{hottbook}
The {Univalent Foundations Program}.
\newblock {\em Homotopy Type Theory: Univalent Foundations of Mathematics}.
\newblock \url{https://homotopytypetheory.org/book}, Institute for Advanced
  Study, 2013.

\end{thebibliography}

\end{document}